\documentclass[journal]{IEEEtran}

\usepackage{amsmath}
\usepackage{amsfonts}
\usepackage{amssymb}
\usepackage{amsthm}
\usepackage{amstext}
\usepackage{mathrsfs}
\usepackage{mathtools}
\usepackage{dsfont}
\usepackage{color}
\usepackage{cite}
\usepackage{times}
\usepackage{ifthen}
\usepackage{latexsym}
\usepackage{verbatim}
\usepackage{psfrag}
\usepackage{bbm}
\usepackage{float}
\usepackage{enumitem}
\usepackage{graphics}
\usepackage{graphicx}
\usepackage{epsfig}
\usepackage{epsf}
\usepackage{algorithm}
\usepackage{algpseudocode}

\usepackage{rotating}
\usepackage{tikz}
\usetikzlibrary{matrix,positioning,decorations.pathreplacing}

\ifCLASSOPTIONcompsoc
  \usepackage[caption=false,font=normalsize,labelfont=sf,textfont=sf]{subfig}
\else
  \usepackage[caption=false,font=footnotesize]{subfig}
\fi

\usepackage{amsmath}
\usepackage{graphicx}
\usepackage{amssymb}
\usepackage[utf8]{inputenc}
\usepackage[english]{babel}
\usepackage[final]{hyperref}

\newcommand{\R}[2]{{R}_{#1 \rightarrow #2}}

\newcommand{\ind}[2]{\mathsf{ind}_{#1} (#2)}

\newcommand{\Fs}{F_{\mathsf{s}}}
\newcommand{\FsI}{F_{\mathsf{s,I}}}
\newcommand{\FsII}{F_{\mathsf{s,II}}}

\newcommand{\MsI}{\bM_{\mathsf{s,I}}}
\newcommand{\CsI}{\bC_{\mathsf{s,I}}}
\newcommand{\MsII}{\bM_{\mathsf{s,II}}}
\newcommand{\CsII}{\bC_{\mathsf{s,II}}}
\newcommand{\bEI}{\bE_{\mathsf{I}}}
\newcommand{\bEII}{\bE_{\mathsf{II}}}

\newcommand{\MsIOv}{\overline{\bM}_{\mathsf{s,I}}}
\newcommand{\MsIUn}{\underline{\bM}_{\mathsf{s,I}}}
\newcommand{\key}{\cQ}
\newcommand{\mat}[2]{#1(#2)}

\newcommand{\lp}{\left(}
\newcommand{\rp}{\right)}
\newcommand{\lb}{\left[}
\newcommand{\rb}{\right]}
\newcommand{\lc}{\left\{}
\newcommand{\rc}{\right\}}
\newcommand{\cnd}{\Big|}

\newcommand{\cEI}{\mathcal{E}_{\mathsf{I}}}
\newcommand{\cEII}{\mathcal{E}_{\mathsf{II}}}

\newtheorem{lm}{Lemma}

\newtheorem{defi}{Definition}
\newtheorem{prop}{Proposition}

\newtheorem{thm}{Theorem}

\newtheorem{remark}{Remark}
\newtheorem{ex}{Example}
\newtheorem{property}{Property}

\newcommand{\bA}{\mathbf{A}}

\newcommand{\bB}{\mathbf{B}}

\newcommand{\bC}{\mathbf{C}}

\newcommand{\bD}{\mathbf{D}}

\newcommand{\bE}{\mathbf{E}}

\newcommand{\bM}{\mathbf{M}}

\newcommand{\bR}{\mathbf{R}}

\newcommand{\bX}{\mathbf{X}}

\newcommand{\bXi}{\mathbf{\Xi}}
\newcommand{\bbJ}{\mathbb{J}}
\newcommand{\bbI}{\mathbb{I}}
\newcommand{\pJ}{\tilde{\mathbb{J}}}

\newcommand{\lbl}[1]{\langle #1 \rangle}

\newcommand{\cA}{\mathcal{A}}

\newcommand{\cG}{\mathcal{G}}
\newcommand{\cH}{\mathcal{H}}
\newcommand{\cI}{\mathcal{I}}
\newcommand{\cJ}{\mathcal{J}}
\newcommand{\cK}{\mathcal{K}}
\newcommand{\cL}{\mathcal{L}}

\newcommand{\cN}{{N}}

\newcommand{\cP}{\mathcal{P}}
\newcommand{\cQ}{\mathcal{Q}}
\newcommand{\cS}{{S}}
\newcommand{\cT}{\mathcal{T}}

\newcommand{\cV}{\mathcal{V}}
\newcommand{\cW}{\mathcal{W}}
\newcommand{\cX}{\mathcal{X}}
\newcommand{\cY}{\mathcal{Y}}

\newcommand{\up}[1]{\overline{#1}}
\newcommand{\down}[1]{\underline{#1}}

\newcommand{\X}[2]{\chi_{#2}^{#1}}

\newcommand{\setl}{\prec}
\newcommand{\nsl}{\lhd}
\renewcommand{\det}[1]{\mathsf{det}\left(#1\right)}

\newcommand{\bxi}{{\mathbf{\Xi}}}
\newcommand{\ubxi}{\up{\mathbf{\Xi}}}
\newcommand{\dbxi}{\down{\mathbf{\Xi}}}
\newcommand{\sbxi}{\widehat{\mathbf{\Xi}}}
\newcommand{\pbxi}{\widetilde{\mathbf{\Xi}}}

\newcommand{\rk}[1]{\text{rank}\left(#1\right)}

\begin{document}
    
\title{Secure Determinant Codes \\for Distributed Storage Systems}
 
 \author{
  \IEEEauthorblockN{
    Adel~Elmahdy\textsuperscript{\textsection}, Michelle~Kleckler\textsuperscript{\textsection} and Soheil~Mohajer%
    \thanks{This work was supported in part by the National Science Foundation under Grants CCF-1617884 and CCF-1749981. A preliminary version of this work was presented in part at the 2019 and 2020 IEEE International Symposium on Information Theory.}
    \thanks{The authors are with the Department of Electrical and Computer Engineering, University of Minnesota, Minneapolis, MN, 55455 USA (e-mail: adel@umn.edu; kleck023@umn.edu; soheil@umn.edu). Corresponding author: S.~Mohajer.}
  }
}

\maketitle
\begingroup\renewcommand\thefootnote{\textsection}
\footnotetext{Equal contribution.}
\endgroup

\begin{abstract}
The information-theoretic secure exact-repair regenerating codes for distributed storage systems (DSSs) with parameters $(n,k=d,d,\ell)$ are studied in this paper. We consider distributed storage systems with $n$ nodes, in which the original data can be recovered from any subset of $k=d$ nodes, and the content of any node can be retrieved from those of any $d$ helper nodes. Moreover, we consider two secrecy constraints, namely, Type-I, where the message remains secure against an eavesdropper with access to the content of any subset of up to $\ell$ nodes, and Type-II, in which the message remains secure against an eavesdropper who can observe the incoming repair data from all possible nodes to a fixed but unknown subset of up to $\ell$ compromised nodes. Two classes of secure determinant codes are proposed for Type-I and Type-II secrecy constraints. Each proposed code can be designed for a range of per-node storage capacity and repair bandwidth for any system parameters. They lead to two achievable secrecy trade-offs, for Type-I and Type-II security. 
\end{abstract}

\begin{IEEEkeywords}
Distributed storage systems, exact-repair regenerating codes, information-theoretic security.
\end{IEEEkeywords}

\section{Introduction}
\IEEEPARstart{W}{ith} the rise in demand and interest for data-driven technology and cloud computing, the size of data and the number of users who wish to access them continue to grow rapidly. This necessitates the need for more efficient as well as secure data storage mechanisms. Recently, the focus of the storage industry has been shifted from central systems to distributed storage systems (DSS). In such systems, data are coded and stored over a set of $n$ storage nodes. These nodes, however, are subject to temporal and permanent failures. Hence, redundancy among the contents of the nodes and node repair mechanisms are essential to retrieve the contents of failed nodes. 

Traditionally, simple coding techniques such as replication-based codes or Reed-Solomon codes have been used to encode data in DSS. While replication-based codes are optimum for node repair, they are very inefficient in terms of storage efficiency. On the other hand, Reed-Solomon codes that require the minimum storage overhead for a given level of reliability impose very heavy network traffic for their repair mechanism. This is due to the fact that they need to download the entire data before the content of a single node (which is a relatively small portion of the data) can be recovered. 

The family of \emph{regenerating codes}, introduced by Dimakis et al.~\cite{dimakis2010network}, strikes a balance between the storage overhead and the cost of node repair (i.e., the communication cost associated with the bandwidth needed to send the repair data) for DSS. In an $(n,k,d)$-DSS with regenerating code parameters $(\alpha,\beta,F)$, the data of size $F$ is encoded into $n$ segments, each is of size $\alpha$ symbols and stored on one storage node. Such systems satisfy two prime properties: (i) data recovery property; and (ii) node repair property. The data recovery property ensures that the original data can be recovered from the content of \emph{any} collection of $k$ nodes. Furthermore, the node repair property guarantees that upon failure of any node, it can be replaced by a new node which, together with the other nodes, maintains the properties of the original system. Such a replacement node can be generated by accessing a collection of $d$ helper nodes and downloading $\beta$ repair symbols from each of them.  Ideally, it is desired to design systems with a small per-node storage $\alpha$ and low per-node repair-bandwidth $\beta$. However, there is a tension between the two parameters that prevents both parameters from being simultaneously minimized. 

There are two types of repair mechanisms: \textit{functional repair} \cite{dimakis2010network}; and \textit{exact repair} \cite{rashmi2011optimal,lin2014unified,ye2017explicit,sasidharan2016explicit,li2017generic,ye2017explicit2,ramkumar2020codes,maturana2020convertible}. In the functional repair, the content of a failed node will be replaced by new content so that the subsequent set of nodes maintains the data recovery property. It is shown in~\cite{dimakis2010network} that there is a fundamental trade-off between $\alpha$ and $\beta$ for the functional repair, which is given~by
\begin{align} 
\label{reg_trade}
    F\leq \sum\nolimits_{i=0}^{k-1}\min\{\alpha,(d-i)\beta\}.
\end{align}
This equation describes a piecewise linear curve in the $\alpha$-$\beta$ plane for a given $F$ and a tuple of system parameters $(n,k,d)$. The extreme points, i.e., the minimum achievable values of $\alpha$ and $\beta$ on the trade-off curve, are called the minimum storage regeneration (MSR) and the minimum bandwidth regeneration (MBR) points, respectively. For functional repair regenerating codes, the upper bound in~\eqref{reg_trade} is achievable, and hence, the resulting storage-bandwidth trade-off is optimum \cite{dimakis2010network}.

On the other hand, in the exact repair, the content of a failed node should be exactly retrieved in the repair process. Exact repair codes are favored in practice because the file recovery process and the meta-data in the system do not change over time. Due to the more stringent constraints, the achievable $(\alpha,\beta)$-region of exact-repair regenerating codes is potentially smaller than that of the functional-repair codes. In contrast to the functional repair, characterization of the optimal storage-bandwidth trade-off for exact-repair regenerating codes remains open for general system parameters.

One of the promising families of exact-repair regenerating codes is the \emph{determinant code}, which is initially proposed for an~$(n,k=d,d)$-DSS in~\cite{elyasi2016determinant, elyasi2019determinant}. A determinant code at \emph{mode} $m\in [d]$ a code with parameters $\alpha=\binom{d}{m}$ and $\beta=\binom{d-1}{m-1}$. It is shown in~\cite{elyasi2016determinant, elyasi2019determinant} that a determinant code at mode $m$ is capable of storing up to $F=m\binom{d+1}{m+1}$ symbols. We refer to Section~\ref{sec:det} for a brief overview of determinant codes. The constraint of $k=d$ is later relaxed in~\cite{elyasi2020cascade}, where the family of \emph{cascade codes} for any tuple~$(n,k,d)$ is introduced. There are~$d$ different determinant codes for a DSS with parameters~$(n,k=d,d)$, each with an individual~$(\alpha,\beta)$ pair. This, together with memory-sharing techniques, leads to a piecewise linear achievable trade-off curve with $d$ corner points. This trade-off includes the optimal points which are only known for specific system parameters, and matches the best known outer bound for linear exact-repair regenerating codes~\cite{elyasi2015linear, prakash2015storage, duursma2015shortened}. 

In many applications, the data stored in a DSS are sensitive and need to be protected against unauthorized or malicious users who wish to access (passive adversary model) or modify (active adversary model) the data. This motivates the idea of information-theoretic secure regenerating codes that guarantee no information leakage about the secure data to an eavesdropper with limited access to the system. Two types of eavesdroppers attacks are studied in the literature \cite{shah2011information}, namely (i) Type-I eavesdropper, who has access to the \emph{contents} of a fixed but unknown set of up to $\ell$ nodes; and (ii) Type-II eavesdropper, who has access to the \emph{incoming repair data} to a fixed but unknown set of up to $\ell$ nodes. We refer to such a system as an~$(n,k,d,\ell)$ system of either Type-I or Type-II. It is worth noting that, due to the node repair mechanism of the system, a Type-II eavesdropper can reconstruct the content of the compromised nodes, and hence is stronger than a Type-I eavesdropper with the same parameter~$\ell$.  The goal of designing secure regenerating codes is to construct codes that ensure the security of the stored data against eavesdroppers, in addition to the data recovery and node repair mechanisms. The performance metric is the \emph{secrecy capacity}, that is, the size of the data that can be securely stored in a code with given parameters $(\alpha, \beta)$. Equivalently, for a given file size, we are interested in characterizing all pairs of $(\alpha,\beta)$, for which a regenerating code with the desired security constraint exists. While it is desired to simultaneously minimize both $\alpha$ and $\beta$, there is a trade-off between the two parameters, and (for an optimum code) one can be decreased only at the cost of increasing the other. Thus, we seek the optimum trade-off between $\alpha$ and $\beta$ for which a certain secrecy capacity can be attained. 

\subsection{Related Works}
Upper bounds on the secrecy capacity for Type-I and Type-II eavesdroppers are presented in~\cite{pawar2011securing}. It is shown that for an $(n,k,d,\ell)$-DSS, the size of secure data in the presence of Type-I or Type-II eavesdroppers, denoted by $\FsI$ and $\FsII$, respectively, must satisfy
\begin{align} 
    \FsII \leq \FsI  \leq \sum\nolimits_{i=\ell}^{k-1}\min\{\alpha,(d-i)\beta\}.
    \label{sec_trade}
\end{align}
The MSR and MBR points on this trade-off have been studied in~\cite{shah2011information}, where information-theoretic security is guaranteed for the MBR point for all feasible $(n,k,d, \ell)$ systems for both Type-I and Type-II security.  Moreover, asymptotically optimal schemes are introduced for the MSR point for feasible $(n,k,d\geq 2k-2, \ell)$ systems under the Type-I security constraint~\cite{shah2011information}. Tandon et al.~\cite{tandon2014new} developed a new upper bound on the secure storage capacity of an $(n, k, d, \ell=1)$-DSS with Type-II eavesdroppers, outperforming the bounds of \cite{pawar2011securing}. Rawat et al.~\cite{rawat2013optimal} proposed tighter bounds on the secrecy capacity for MSR codes, and provided secure coding schemes that achieve their bounds for both Type-I and Type-II eavesdroppers with at most $\ell=2$ compromised nodes, under the assumption that the Type-II attacked nodes are among the systemic nodes.   Goparaju~et~al.~\cite{goparaju2013data} improved these bounds and proved that, under the class of \emph{linear} regenerating codes, the construction proposed in~\cite{rawat2013optimal} provides an optimum MSR code for any number of compromised nodes. However, for the codes in~\cite{goparaju2013data}, the repair process is only guaranteed for the systematic nodes, and hence the set of Type-II compromised nodes is limited to subsets of the systematic ones.

Tandon et al. \cite{tandon2016toward} characterized the secure trade-off region of an $(n, k, d)$-DSS for $n \leq 4$ and $\ell < k$ in the presence of Type-I and Type-II adversaries. Moreover, those results are extended for an $(n, k=n\hspace{-1pt}-\hspace{-1pt}1, d=n\hspace{-1pt}-\hspace{-1pt}1, \ell=n\hspace{-1pt}-\hspace{-1pt}2)$ DSS. The first trade-off curve with \emph{multiple} corner points was precisely characterized in \cite{shao2017tradeoff} for the a~$(7,6,6)$ system, which is secured against $\ell=1$ Type-II eavesdroppers. The Type-II secrecy capacity for some range of specific system parameters is characterized in \cite{ye2017rate}. It is shown in~\cite{ye2019secure} that the trade-off reduces to a single point when the system parameters are $(n=d+1,k=d,d,\ell)$, and $\ell \geq \Bigl\lceil{\frac{d-1}{4}}\Bigr\rceil$.

Recently, Kruglik~\cite{kruglik2020secure} investigated the security issues of MBR array codes under a special type of eavesdroppers that can attack all storage nodes in a distributed storage system, but only access a small number of symbols stored in each node. The author proposed an explicit construction of MBR array codes that is secure against such an eavesdropper. An upper bound on the secure storage capacity of such systems is established and shown to be tight by providing a coding scheme that achieves the bound. Gulcu~\cite{gulcu2020secure} considered the problem of repairing a node in a secure DSS, which is developed based on the Reed-Solomon codes. The author proposed a secure node repair algorithm that performs near-optimal in terms of the bandwidth under a low-rate Reed-Solomon code. Due to the increasing storage requirement for blockchains, coding theoretic techniques have been proposed to alleviate the storage cost and the bootstrap cost that would help more miners enter the market. Gadiraju~et~al.~\cite{gadiraju2020secure} proposed a sharding protocol that is based on exact repair secure regenerating codes. It is shown that the proposed protocol is storage and bandwidth efficient for a single node failure.  Moreover, an equivalence between the process of bootstrapping a node and repairing a failed node is established to demonstrate that the bootstrap cost is low as compared to uncoded sharding. Liang~et~al.~\cite{liang2020secure} designed a secure data storage system and a recovery scheme for blockchain-based industrial networks. The proposed regeneration code exhibits simple coding characteristics and excellent capability of local repair. Furthermore, experiments show that the proposed scheme reduces the repair overhead of local code in data storage nodes and enhances the data integrity in the blockchain.

\subsection{Main Contributions}
In this paper, we generalize $(n,k=d,d)$ determinant codes \cite{elyasi2019determinant}, which is a class of (non-secure) optimum exact-repair regenerating codes, to achieve information-theoretic security in the presence of Type-I or Type-II eavesdroppers. We summarize the main results of this paper as follows:
\begin{itemize}[leftmargin=*]
    \item We provide explicit code constructions with a fairly small field size for Type-I and Type-II secure determinant codes for an $(n,k=d,d)$-DSS and for any number of compromised nodes $1\leq \ell \leq k$. We characterize the achievable trade-offs, that consist of $k=d$ corner points for  Type-I security, and at most $k=d$ corner points under Type-II security.
    \item We characterize the number of linearly independent variables observed by both Type-I and Type-II eavesdroppers, to determine the number of random keys required to guarantee security.
    \item We prove that the proposed code constructions satisfy three properties: (i) data recovery property, (ii) node repair property; and (ii) Type-I or Type-II security constraints.
    \item We prove the optimality of the proposed constructions among all determinant-based codes. More precisely, we show that the proposed Type-I and Type-II secure determinant codes store the maximum secure file size $\Fs$ that can be securely stored in a determinant code.
\end{itemize}
A summary of the main results of this paper has been presented in \cite{kleckler2019secure} for Type-I secure determinant codes, and \cite{kleckler2020secure} for Type-II secure determinant codes. This paper presents complete proofs of all results, introduces new results about the optimality of the code construction of Type-II secure determinant codes, and provides numerous illustrative examples to compare between Type-I and Type-II secure code constructions, compared to the non-secure version of determinant codes.

\subsection{Notation} 
For integers $a$ and $b$ we use $[a:b]$ to denote the set of integers $\{a,a+1,\ldots, b\}$, and $[b] = [1 : b]$. Note that $[a:b]$ is an empty set if $a>b$. For integers $0\leq a \leq b$, we define $\binom{b}{a}=\frac{b!}{a! (b-a)!}$. Furthermore, we define $\binom{b}{a}=0$, if $a>b$ or $a<0$. We use lowercase letters (e.g., $n$) to refer to (real and finite field) numbers, and random variables and random vectors are indicated by capital letters (e.g. $S$). Calligraphic letters (e.g., $\cI$) denote sets, and $|\cI|$ denotes the cardinality of~$\cI$. Boldface capital letters (e.g. $\bA$) denote matrices, and~$\bA^{-1}$ denotes the inverse of matrix $\bA$. For a matrix $\bA$, $\bA(i,j)$ denotes the matrix entry at row $i$ and column $j$. Moreover, $\bA(i,:)$ and $\bA(:,j)$ denote the $i$th row and $j$th column of matrix $\bA$, respectively. Furthermore, for sets $\cI$ and $\cJ$, the submatrix of $\bA$ obtained by  rows $i\in\cI$ and columns $j\in\cJ$ is denoted by $\bA(\cI,\cJ)$. For a set $\cI$ and a set member $x \in \cI$, we define $\text{ind}_\cI(x) \coloneqq |\{y\in\cI : y \leq x \}|$. Moreover, the largest and smallest entries of the set $\cI$ are denoted by $\max \cI$ and $\min \cI$, respectively. For sets $\cI$ and $\cJ$, $\cJ\subseteq\cI$ means  that $\cJ$ is a subset of $\cI$. Also, we define the set difference as $\cI\setminus\cJ\coloneqq\{x\in \cI: x\notin \cJ\}$. Furthermore, when $|\cI|=|\cJ|$, we say $\cI$ is lexicographically smaller than $\cJ$ and denote it by $\cI \prec \cJ$, if $\min \cI\setminus \cJ < \min \cJ\setminus \cI$. For example, $\{1,2,5\}\prec\{1,3,4\}$. All symbols in the code construction are assumed to be elements of a Galois field $\mathbb{F}_q$ for some prime power  $q$, and the entropy function $H(\cdot)$ is computed in base~$q$.

\subsection{Paper Outline}
The remainder of the paper is organized as follows. We first present the problem formulation and the main results of this work in Section~\ref{sec:probForm_mainRes}. The code construction of the non-secure determinant codes is reviewed in Section~\ref{sec:det}. In Section~\ref{sec:typeI}, we discuss Type-I secure determinant codes. More specifically, the code construction is proposed in Section~\ref{sec:codeConstruct_I}, an illustrative example is provided in Section~\ref{sec:codeEx_I}, the proposed construction of determinant codes is proved to satisfy Type-I security constraint in Section~\ref{sec:achvProof_I}, and finally, the optimality of the proposed Type-I secure code construction for determinant codes is shown in Section~\ref{sec:converseProof_I}. Section~\ref{sec:typeII} is dedicated to Type-II secure determinant codes, where the code construction is presented in Section~\ref{sec:codeConstruct_II} followed by an illustrative example in Section~\ref{sec:codeEx_II}, the security property of the proposed construction is proved in Section~\ref{sec:achvProof_II}, and the optimality of the proposed code construction for determinant codes is established in Section~\ref{sec:converseProof_II}. Finally, the paper is concluded in Section~\ref{sec:conclusion}. The paper has six appendices, where the proofs of some of the technical claims are presented. 

\section{Problem Formulation and Main Results}
\label{sec:probForm_mainRes}
\subsection{Problem Formulation}
We study the fundamental trade-off between the per-node storage and repair bandwidth for secure distributed storage systems under both Type-I and Type-II eavesdroppers. 

An exact-repair regenerating code with system parameters $(n,k,d )$ and code parameters $(\Fs, \alpha, \beta)$ maps a secure message $\cS$ of size $\Fs$ symbols (i.e., $H(S)=\Fs$) to $n$ codewords, namely $\cN_1, \cN_2, \dots, \cN_n$, each of size $H(\cN_i)\leq \alpha$ symbols for $i\in [n]$. The codewords should satisfy the following properties: 
\begin{enumerate}
    \item \textbf{Data Recovery}: The original file can be reconstructed from the content of any set of $k$ nodes, that is, 
    \begin{align}
    H (\cS | \{\cN_i: i\in \cK\}) = 0, \:\text{ for any } 
    \cK\subseteq [n] \text{ and } |\cK|=k.
    \label{eq:recpvery}
    \end{align}
    \item \textbf{Exact Node Repair}: 
    Whenever a node $f \in [n]$ fails and becomes inaccessible, it can be repaired and its content $\cN_f$ can be \emph{exactly} reconstructed from the repair data of size at most $\beta$ symbols received from any collection of $d$ helper nodes. More precisely, for every failed node $f$,  every set of helper nodes $\cH\subseteq [n]\setminus\{f\}$ with $|\cH|=d$, and every helper node $h\in \cH$,  there exists repair data encoders that generate\footnote{When helper nodes in $\cH$ contribute to repair a failed node $f$, the repair data sent from a helper node $h\in \cH$ to $f$ may depend on the identity of other contributing helper nodes, and it is more appropriate to be denoted by $\R{h}{f}^\cH$. However, in this paper, we are using the helper independent construction for determinant codes~\cite{elyasi2019determinant}, in which $\R{h}{f}^\cH$ does not depend on $\cH$. Hence, for ease of notation, we use $\R{h}{f}$ to refer to the repair data sent from $h$ to $f$.} $\R{h}{f}$, the  repair data\footnote{We define $\R{h}{h}=\varnothing$ to be a dummy variable with zero entropy for convenience.} that are sent from the helper node $h$ to the failed node $f$, that satisfy
    \begin{align}
        &H(\R{h}{f}| \cN_h) = 0, \:\text{ for } h\in [n]\setminus\{f\},\nonumber
        \\
        &H(\R{h}{f}) \leq \beta, \:\text{ for } h\in [n]\setminus\{f\},\label{eq:repair}\\
        &H(\cN_f | \{\R{h}{f}\!: h\!\in\! \cH\}) \!=\! 0, \:\text{ for } \cH \!\subseteq\![n]\!\setminus\! \{f\},\: |\cH|\!=\!d.\nonumber
    \end{align}
\end{enumerate}
Next, we explain the information-theoretic secrecy constraints for Type-I and Type-II security.
\begin{itemize}
    \item \textbf{Type-I Security}:
    An $(n,k,d)$ exact regenerating code is called an $(n,k,d,\ell)$ \emph{Type-I secure code} if an eavesdropper with access to the \emph{content} of an arbitrary subset of at most $\ell$ nodes (with $\ell<k$) cannot learn \textit{anything} about the secure message $\cS$. That is, for any set of compromised nodes $\cL\subseteq [n]$ accessed by the eavesdropper with $|\cL|\leq\ell$, we define the set  eavesdropper's observed variables by $\cEI(\cL) = \{\cN_i: i\in \cL\}$, and have 
    \begin{align} 
        I(\cS;\cEI(\cL))=0, \qquad 
        \forall \cL \subseteq [n] \text{ and } |\cL|\leq\ell.
        \label{req_mut_I}
    \end{align}
    \item \textbf{Type-II Security}:
    An $(n,k,d)$ exact regenerating code is called an $(n,k,d,\ell)$ \emph{Type-II secure code} if an eavesdropper with access to all the \emph{incoming repair data from all possible helpers to a fixed but unknown subset of at most $\ell$ compromised nodes} (with $\ell \leq k$) can not learn \textit{any} information about the secure message $\cS$. That is, for any set of compromised nodes  ${\cL\subseteq [n]}$ accessed by the eavesdropper with $|\cL|\leq\ell$, and ${\cEII(\cL) = \{\R{h}{i}: i\in \cL, h\in [n]\setminus\{i\}\}}$, we have 
    \begin{align} 
        I(\cS;\cEII(\cL))=0, \qquad 
        \forall \cL \subseteq [n] \text{ and } |\cL|\leq\ell.
        \label{req_mut_II}
    \end{align}
\end{itemize}
We denote the maximum size of a secure message that can be stored in a Type-I secure DSS by $\FsI$. Similarly, the maximum size of a secure message that can be stored in a DSS with Type-II secrecy constraint is denoted by $\FsII$. 

\begin{remark}
    The secrecy constraints in \eqref{req_mut_I} and \eqref{req_mut_II} imply security in an information-theoretic sense: the eavesdropper with unbounded computational power, unlimited amount of time, and full knowledge of the underlying code construction would not be able to learn anything about the secure message from the observation.
\end{remark}

\begin{remark}
    The Type-II secrecy constraint in \eqref{req_mut_II} is stronger than the Type-I secrecy constraint in \eqref{req_mut_I}. This is due to the fact that the content of a node can be exactly retrieved from the repair data coming from any $d\leq n-1$  nodes. More precisely, we have $H(\cEI(\cL)|\cEII(\cL), \cS) = H(\cEI(\cL)|\cEII(\cL))=0$ for every $\cL\subseteq [n]$. Therefore, 
    \begin{align}
        I(\cS;\cEI(\cL)) 
        &= H(\cS) - H(\cS|\cEI(\cL))
        \nonumber\\
        &\leq H(\cS) - H(\cS|\cEI(\cL),\cEII(\cL))
        \nonumber\\
        &= H(\cS) - H(\cS|\cEII(\cL))\nonumber\\
        &\qquad - H(\cEI(\cL)|\cEII(\cL), \cS) + H(\cEI(\cL)|\cEII(\cL))
        \nonumber\\
        &= H(\cS) - H(\cS|\cEII(\cL))
        \nonumber\\
        &= I(\cS;\cEII(\cL)). 
        \nonumber
    \end{align} 
    This implies that Type-II security is more stringent compared to Type-I security, and hence $\FsII \leq \FsI$ for any pair of $(\alpha, \beta)$ and any set of system parameters $(n,k,d)$. 
\end{remark}   

\subsection{Main Results}
For an $(n,k,d,\ell)$ distributed storage system with code parameters $(\alpha,\beta)$, the main goal is to characterize the maximum $\Fs$ for which there exists a secure $(\Fs, \alpha, \beta)$ exact-repair regenerating code against a particular type of eavesdroppers. Since parameters $\alpha$, $\beta$, and $\Fs$ scale linearly together, this is equivalent to characterizing the trade-off between the normalized parameters $\alpha/\Fs$ and $\beta/\Fs$. Characterization of this trade-off is an open problem for general $(n,k,d,\ell)$ systems. In this paper, we focus on the systems with $k=d$, and introduce a family of secure codes based on determinant codes operating at different trade-off points. This leads to an achievable trade-off, which establishes an upper bound on the optimum trade-off. We also prove that this bound is tight for determinant codes, that is, the proposed code constructions offer the maximum secure capacity within the class of determinant codes. In what follows, we present the main results of this paper on developing secure determinant codes for Type-I and Type-II security.

\subsubsection*{\textbf{Type-I Security}}
Let $\FsI$ be the size of the Type-I secure message. The following theorem characterizes the set of achievable tuples for a family of Type-I secure determinant codes. 
\begin{thm}
    For an $(n,k=d,d,\ell)$ distributed storage system with Type-I security constraint, the tuples in the convex hull of $\left\{\left(\alpha^{(m)}, \beta^{(m)}, \FsI^{(m)} \right): m \in [d]\right\}$ with 
    \begin{align}
        \left\{
        \begin{array}{l}
            \alpha^{(m)} = \binom{d}{m},\\
            \beta^{(m)} = \binom{d-1}{m-1},\\
            \FsI^{(m)} = (d-\ell)\binom{d}{m}-\binom{d}{m+1} + \binom{\ell}{m+1},
        \end{array}
        \right.
       \label{eq:thm_achv_I}
    \end{align}
    for $m\in [d]$ are achievable through an explicit and efficient code construction with a fairly small field size. The symbols in the code construction are elements of a Galois field $\mathbb{F}_q$, where $q$ can be any prime power satisfying $q > n$.
\label{thm:achv_I}
\end{thm}
We present a code construction with the parameters of  Theorem~\ref{thm:achv_I} in Section~\ref{sec:codeConstruct_I}. An illustrative example is then given in Section~\ref{sec:codeEx_I}. Finally, the achievability proof is provided in Section~\ref{sec:achvProof_I}.

\subsubsection*{\textbf{Type-II Security}}
Let $\FsII$ be the secrecy capacity of the Type-II  regenerating code. The following theorem characterizes the set of achievable tuples for a family of Type-II secure determinant codes.  
\begin{thm}
    For an $(n,k=d,d,\ell)$ distributed storage system with a Type-II security constraint, all tuples in the convex hull of $\left\{\left(\alpha^{(m)}, \beta^{(m)}, \FsII^{(m)} \right): m \in [d-\ell]\right\}$ with
    \begin{align}
        \left\{
        \begin{array}{l}
            \alpha^{(m)} = \binom{d}{m},\\
            \beta^{(m)} = \binom{d-1}{m-1},\\
            \FsII^{(m)} = m\binom{d-\ell+1}{m+1},
        \end{array}
        \right.
      \label{eq:thm_achv_II}
    \end{align}
    for $m\in[d]$ are achievable through an explicit and efficient code construction with a fairly small field size. The symbols in the code construction are elements of a Galois field $\mathbb{F}_q$, where $q$ can be any prime power satisfying $q > n$.
\label{thm:achv_II}
\end{thm}
The code construction with the parameters of Theorem~\ref{thm:achv_II} is presented in Section~\ref{sec:codeConstruct_I}. An illustrative example is then given in Section~\ref{sec:codeEx_II}. Finally, the achievability proof is provided in Section~\ref{sec:achvProof_II}. 

Figure~\ref{fig:type_I_ex} depicts the secure storage-bandwidth trade-off curves for different numbers of Type-I eavesdroppers, while Figure~\ref{fig:type_II_ex} captures the secure storage-bandwidth trade-off curves. Both figures are for a system with $k=d=15$, and each curve shows the trade-off for one value of $\ell\in\{0,1,2,3\}$. Note that when $\ell=0$, there are no security constraints, and the trade-off curves reduce to that in \cite{elyasi2019determinant}.  

Figures~\ref{fig:type_I_II_ell1}~and~\ref{fig:type_I_II_ell2} compare the secure storage-bandwidth trade-off curves for an $(n,k,d,\ell)=(n,30,30,\ell)$-DSS for Type-I and Type-II secure determinant codes for $\ell=1$ and $\ell=2$, respectively. One interesting observation is that the MBR points (i.e., at $m=1$) for Type-I and Type-II secure determinant codes are identical for any $\ell \leq k$, and the corresponding maximum secure file size is given by
\begin{align*}
	\FsI^{(1)} = \FsII^{(1)} = \frac{1}{2} (d-\ell+1) (d-\ell).
\end{align*}

\begin{figure}
	\centering
	\subfloat[]{
		\includegraphics[width=0.4\textwidth]{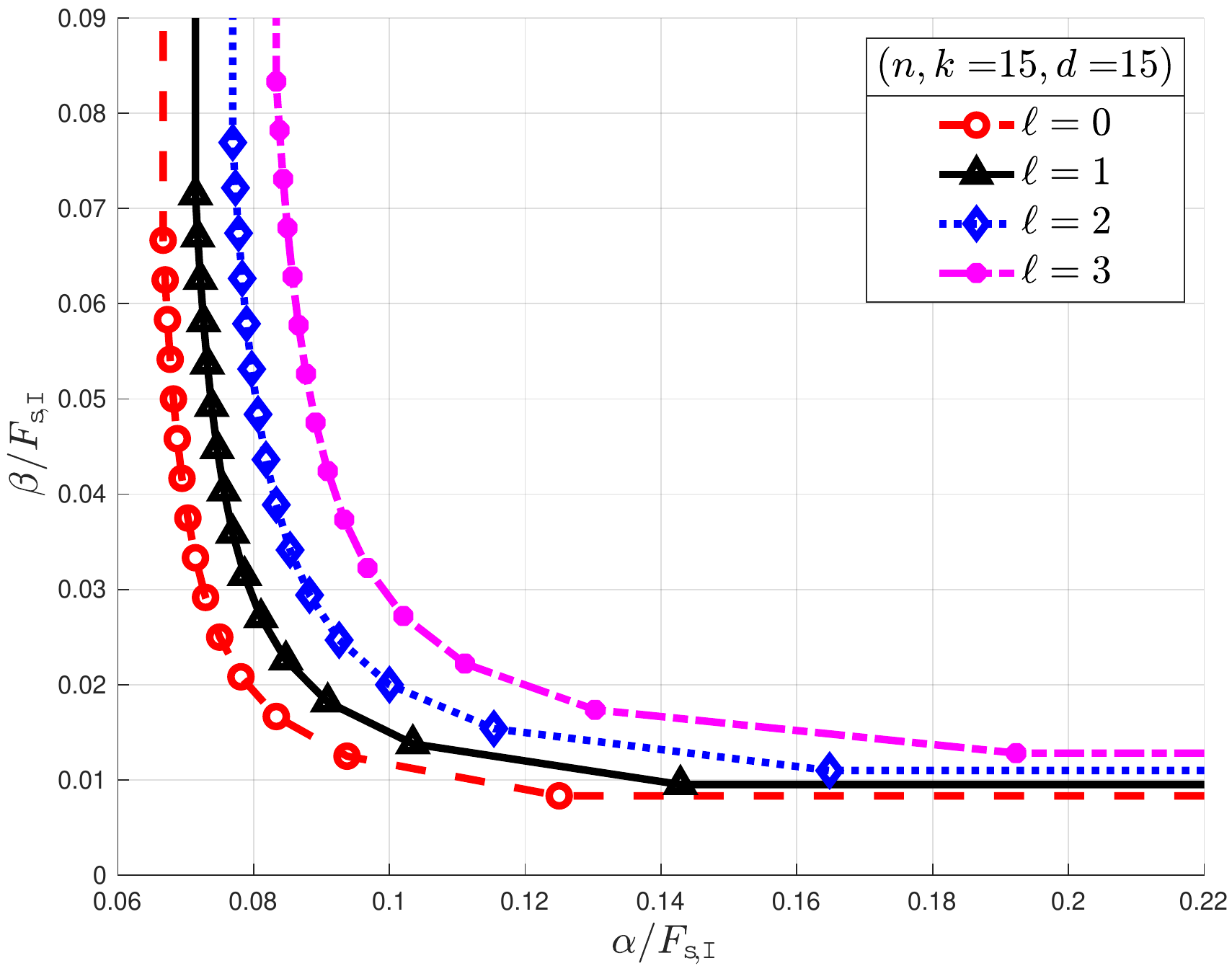}
		\label{fig:type_I_ex}
	}
	\hspace{10mm}
	\subfloat[]{
	\includegraphics[width=0.4\textwidth]{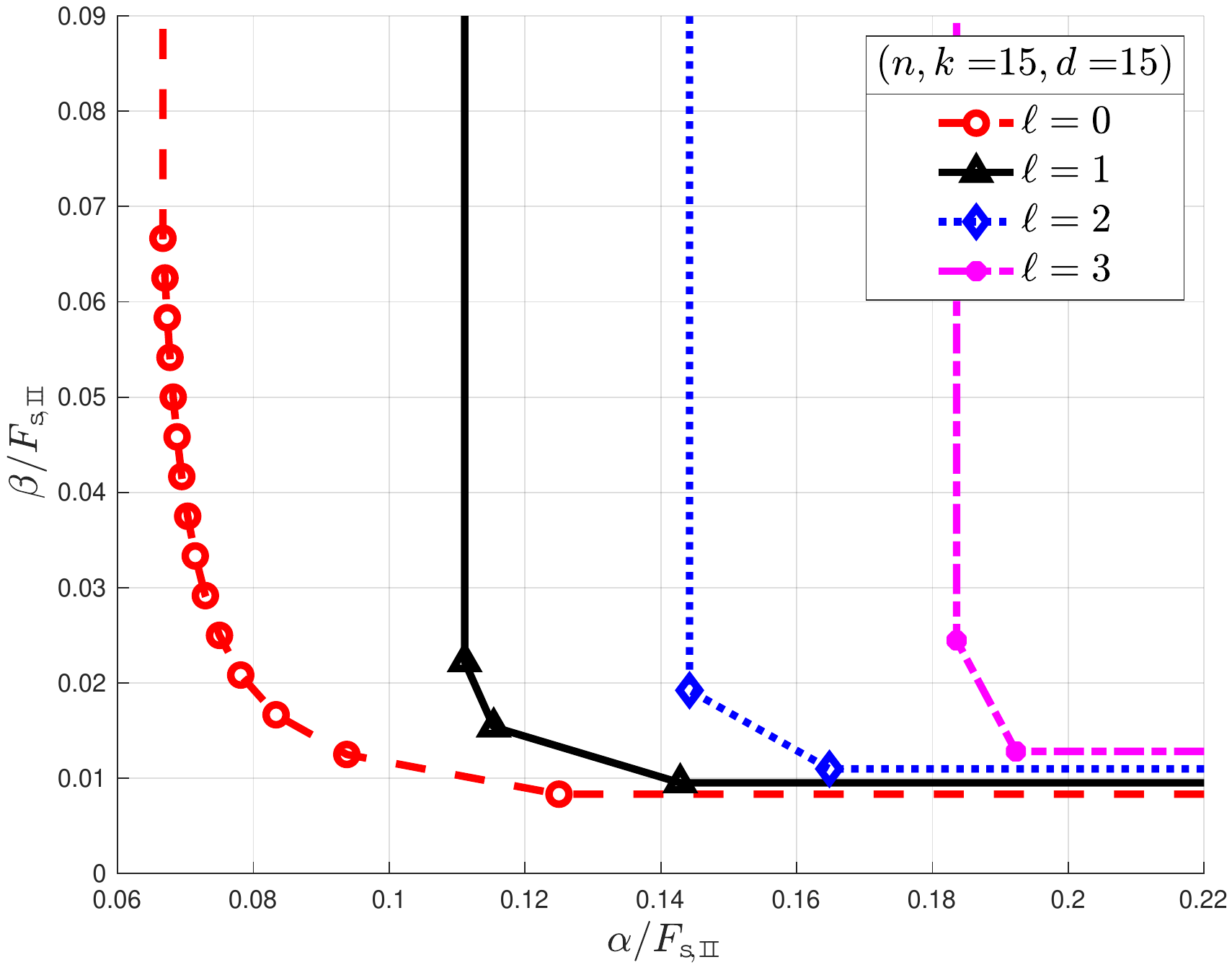}
	\label{fig:type_II_ex}
	}
	\caption{
		(a) The trade-off curve of Type-I secure determinant codes with parameters $(n,k,d,\ell)=(n,15,15,\ell)$ for $\ell \in \{0,1,2,3\}$.
		(b) The trade-off curve of Type-II secure determinant codes with the same system parameters.  
	}
\end{figure}

\begin{figure}
	\centering
	\subfloat[]{
		\includegraphics[width=0.4\textwidth]{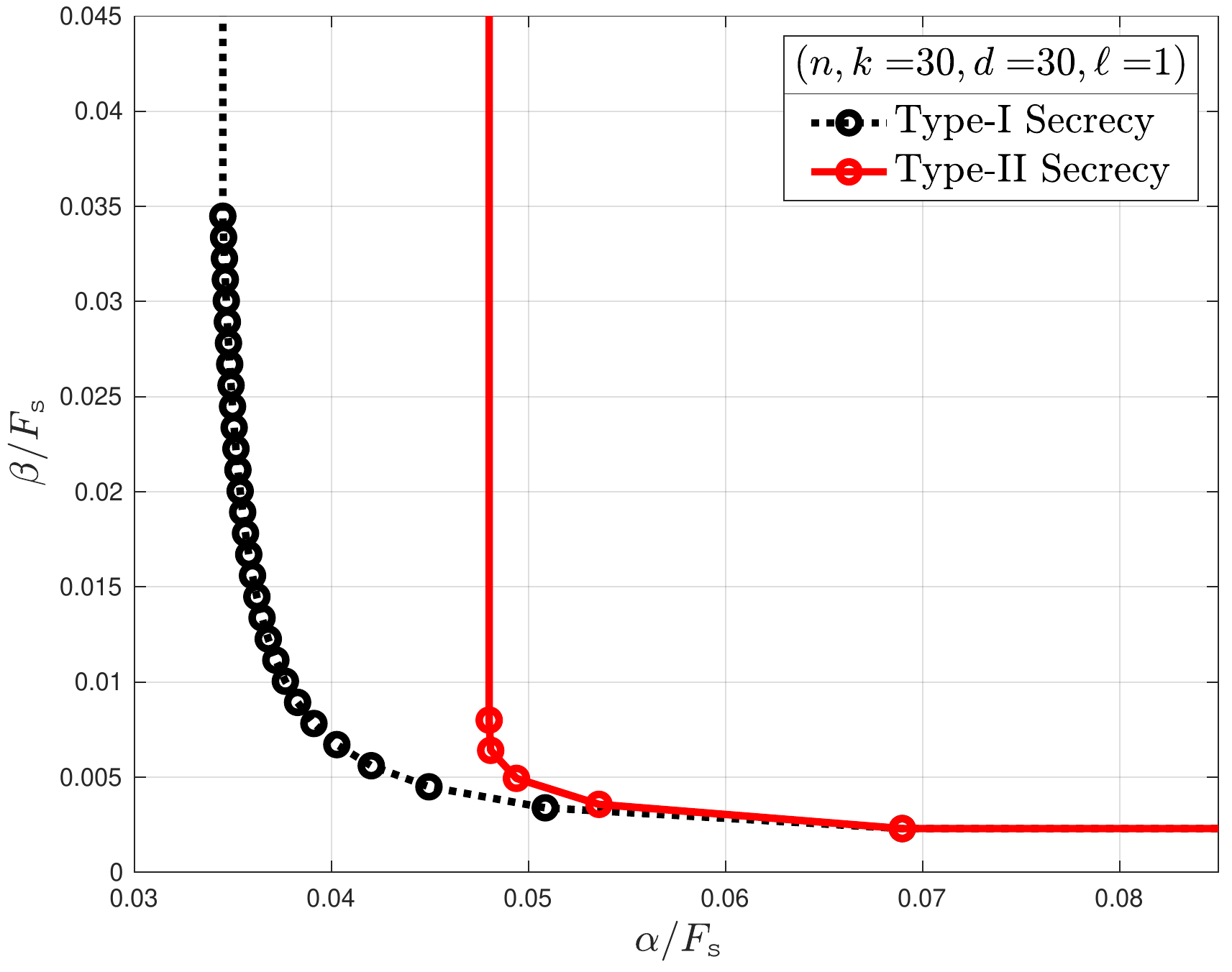}
		\label{fig:type_I_II_ell1}
	}
	\hspace{10mm}
	\subfloat[]{
	\includegraphics[width=0.4\textwidth]{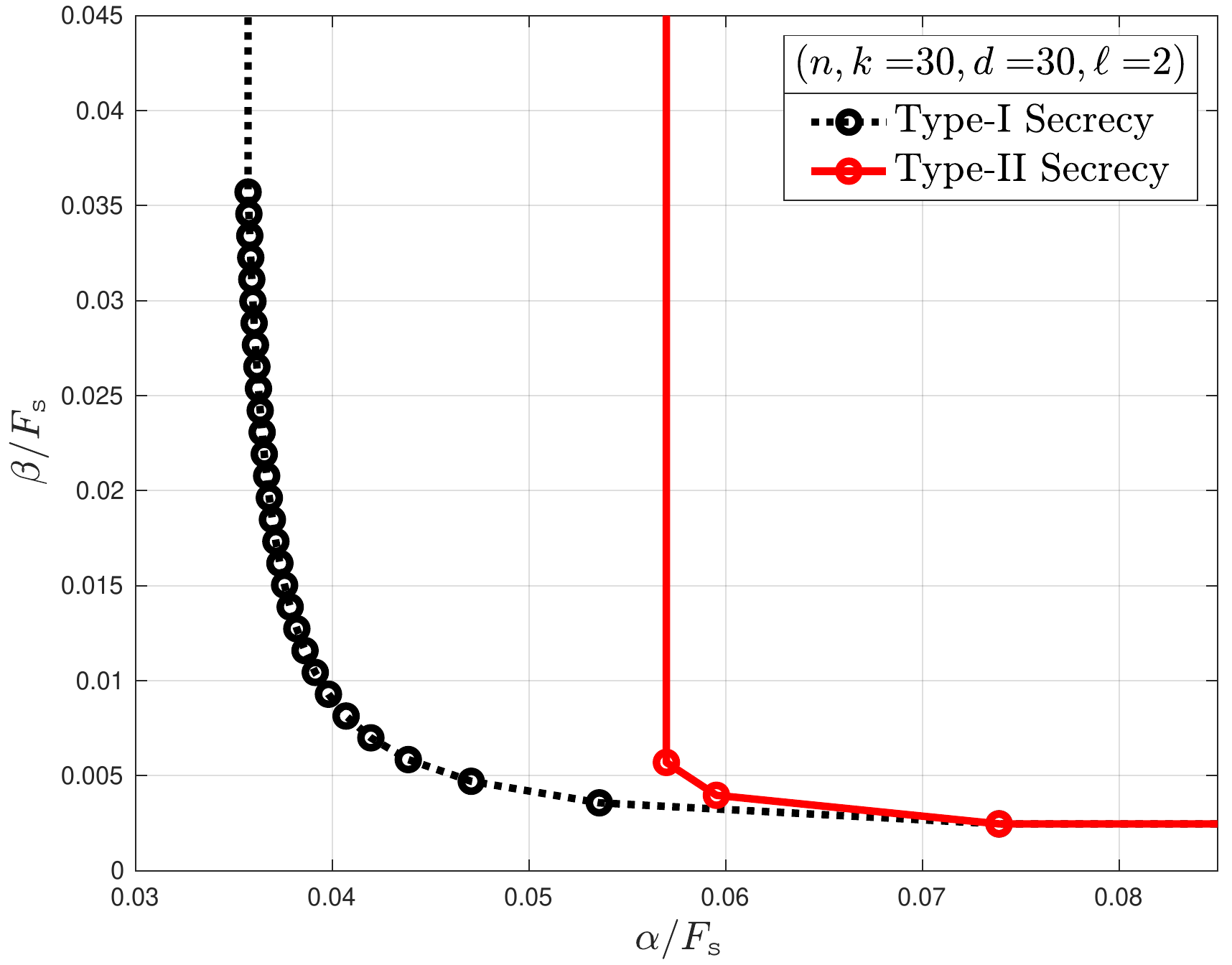}
	\label{fig:type_I_II_ell2}
	}
	\caption{
		(a) The trade-off curve of Type-I and Type-II secure determinant codes with parameters $(n,k,d,\ell)=(n,30,30,1)$.
		(b) The same trade-off curves for determinant codes with parameters $(n,k,d,\ell)=(n,30,30,2)$.
	}
\end{figure}

\subsection{Properties of Secure Determinant Codes}
The next property shows that the file size $\FsI$ in Theorem~\ref{thm:achv_I} is the maximum secure capacity that one can achieve using determinant codes.
\begin{property}
    For an $(n,k=d,d,\ell)$ Type-I secure determinant code operating at mode $m\in [d]$ with parameters $\alpha=\binom{d}{m}$ and $\beta=\binom{d-1}{m-1}$, the maximum secure file size is upper bounded by 
    \begin{align*}
      H(\cS) \leq \FsI^{(m)} = (d-\ell)\binom{d}{m}-\binom{d}{m+1} + \binom{\ell}{m+1}.
    \end{align*}
    \label{thm:converse_I}
\end{property}
The proof of Property~\ref{thm:converse_I} is provided in Section~\ref{sec:converseProof_I}.

\begin{remark}
    Theorem \ref{thm:achv_I} and Property \ref{thm:converse_I} prove the optimality of the proposed Type-I secure code construction under the constraint that the code belongs to the family of determinant codes. It is an open problem to prove that the proposed code construction is optimal across all Type-I secure exact-repair regenerating DSS codes with parameters.
\end{remark}

The next property shows that the  $\FsII$ introduced in  Theorem~\ref{thm:achv_II} is the maximum file size one can achieve using determinant codes.
\begin{property}
    For an $(n,k=d,d,\ell)$ Type-II secure determinant code operating at mode $m\in [d]$ with parameters $\alpha=\binom{d}{m}$ and $\beta=\binom{d-1}{m-1}$, the maximum secure file size is upper bounded by 
    \begin{align*}
        H(\cS) \leq \FsII^{(m)} = m\binom{d-\ell+1}{m+1}.
    \end{align*}
    \label{thm:converse_II}
\end{property}
The proof of Property~\ref{thm:converse_II} is provided in Section~\ref{sec:converseProof_II}.
\begin{remark}
    Theorem \ref{thm:achv_II} and Property \ref{thm:converse_II} prove the optimality of the proposed Type-II secure code construction under the constraint that the code belongs to the family of determinant codes. It is an open problem to prove that the proposed code construction is optimal across all Type-II secure exact-repair regenerating DSS codes with parameters.
\end{remark}
It should be noted that even though the achievable secrecy trade-off defined in Theorem~\ref{thm:achv_II} is characterized by $d$ points (enumerated by $m\in [d]$), the region may indeed have fewer corner points (or Pareto points \cite{ye2019secure}). This is due to the fact that many of the points introduced in~\eqref{eq:thm_achv_I} are interior points, i.e., they lie in the convex hull of other corner points.  Figure~\ref{fig:region} depicts the achievable region for a system with ${(n,k=d=10,\ell=2)}$. The achievable (normalized) trade-off is only characterized by $2$ corner points, associated with $m=1$ and $m=2$. The code associated with $m=3$ (and all other $m>3$) offers an achievable point that belongs to the convex hull of the points for $m=1$ and $m=2$. The following property of Theorem~\ref{thm:achv_II} characterizes the number of Pareto points of the achievable trade-off.
\begin{property}
    The achievable trade-off of the proposed code construction in Theorem~\ref{thm:achv_II} has exactly $t$  Pareto (extreme) points, where $t$ is the largest integer satisfying
    \begin{align}
        t < \frac{ \sqrt{1+4\ell (d+1)}-1}{2\ell}. 
        \label{eq:con:pareto}
    \end{align}
\label{cor:num_pareto}
\end{property}
The proof of Property~\ref{cor:num_pareto} is presented in Appendix~\ref{app:pareto}. 

\begin{figure}[!t]
    \centering
    \includegraphics[width=0.40\textwidth]{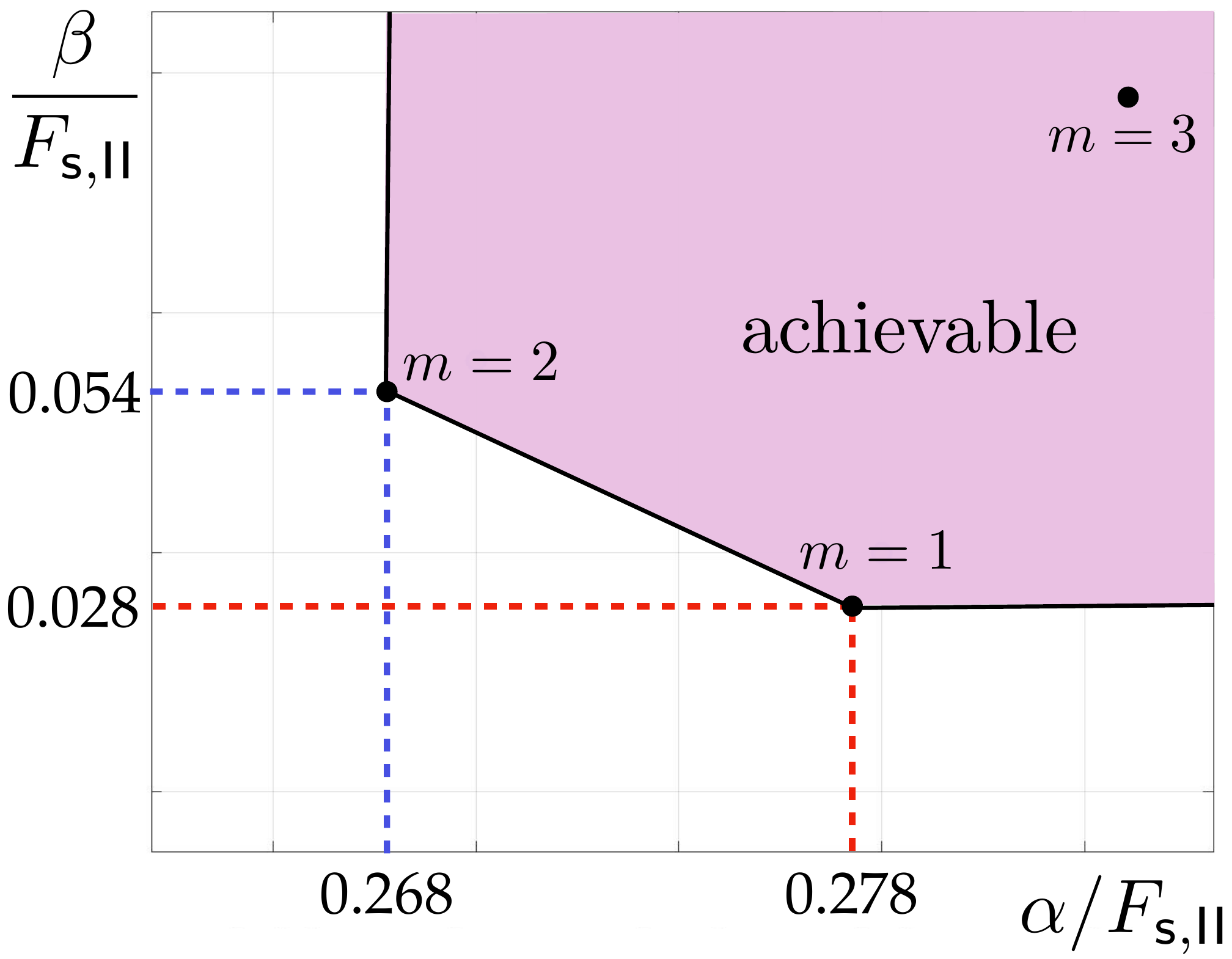}
    \caption{The trade-off curve of a Type-II determinant code with parameters $(n,k,d,\ell)=(n,10,10,2)$.
    }
    \label{fig:region}
\end{figure}

\subsection{Comparison Against other Secure Regenerating Codes}
\label{sec:comp}
In this section, we present a brief comparison between the performance of the existing secure exact-regenerating codes in the literature and the code constructions reported in Theorem~\ref{thm:achv_I} and Theorem~\ref{thm:achv_II}. 
\begin{itemize}
    \item A Type-I secure code construction for the MBR point is presented in~\cite{shah2011information} that achieves the cut-set bound, and hence is optimum. When $d=k$, the parameters of the proposed code in~\cite{shah2011information} satisfy ${\alpha=d\beta}$ and ${\FsI = \left(d^2 - \binom{d}{2}\right)\beta - \left(\ell d - \binom{\ell}{2}\right)\beta}$. It is worth noting the code parameters in Theorem~\ref{thm:achv_I} for mode $m=1$ satisfy $\alpha^{(1)}=d$, $\beta^{(1)}=1$, and ${\FsI^{(1)}=(d-\ell) d - \binom{d}{2}+\binom{\ell}{2}}$, which coincide with those of~\cite{shah2011information}, and thus, the proposed determinant code is optimum. 
    \item It is shown in~\cite{tandon2014new} that the secrecy capacity of any \mbox{Type-II} code with parameters $(n,k,d,\ell=1)$ satisfies 
    \[
    \FsII\leq \frac{k-1}{4}\alpha + \frac{(k-1)(3d-2k)}{4}\beta.
    \]
    For the regime of $k=d$, this bound reduces to ${\FsII \leq (d-1)(\alpha + d\beta)/4}$. For the code parameters of Theorem~\ref{thm:achv_II} with $\ell=1$, we have 
    \begin{align*}
        \frac{d-1}{4}\left(\alpha^{(m)} + d\beta^{(m)}\right) 
        &=\frac{d-1}{4}\left(\binom{d}{m} + d\binom{d-1}{m-1}\right)  
        \nonumber\\
        &= \frac{(d-1)(m+1)}{4}\binom{d}{m} 
        \nonumber\\
        &= \frac{(d-1)(m+1)^2}{4(d-m)m}m\binom{d}{m+1} 
        \nonumber\\
        &\geq m\binom{d}{m+1} = \FsII^{(m)}.
    \end{align*}
    Here, the inequality hold since 
    \begin{align*}
        (d-1)(m+1)^2 &= 4(d-m)m + d(m-1)^2 \\
        &\qquad+ (m-1)(3m+1)\\ 
        &\geq 4(d-m)m,
    \end{align*}
    and hence, the proposed code satisfies the upper bound. Moreover, for $m=1$ the bound is satisfied with equality, and hence the proposed MBR codes are optimum. 
    \item In~\cite[Theorem~1]{tandon2016toward}, the secure capacity is characterized for a DSS with parameters$(n,k=2,d,\ell=1)$. For the sake of comparison, we need to set $d=k=2$, where the result of~\cite{tandon2016toward} reduces to $\FsI= \min(\alpha, \beta)$ and ${\FsII=\min(\alpha/2, \beta)}$. These capacities match the achievable region obtained Theorem~\ref{thm:achv_I} and Theorem~\ref{thm:achv_II}, and hence our codes are optimum for the parameters of interest.  
    \item For parameters ${(n=d+1, k=d,d, \ell=d-1)}$, the secrecy capacity is characterized in~\cite[Theorem~2]{tandon2016toward}, and given by $\FsI=\min(\alpha, \beta)$ and $\FsII=\min(\alpha/d,\beta)$. For this set of parameters, $\FsI^{(m)}$ in Theorem~\ref{thm:achv_I} reduces to 
    \[
    \FsI^{(m)} = \binom{d}{m} - \binom{d}{m+1} + \binom{d-1}{m+1} = \binom{d-1}{m-1}. 
    \]
    In this regime, the achievable tuple at mode $m=d$ is $\left(\alpha^{(d)},\beta^{(d)}, \FsI^{(d)}\right) = (1,1,1)$, which dominates all other achievable tuples and fully characterizes the secrecy capacity region. Moreover, for a Type-II eavesdropper, Theorem~\ref{thm:achv_II} provides a single achievable tuple $\left(\alpha^{(1)},\beta^{(1)}, \FsII^{(1)}\right) = (d,1,1)$ which matches the result of~\cite{tandon2016toward}. 
    \item An upper bound for the Type-II secrecy capacity of an $(n,k,d,\ell)$ system is presented in ~\cite[Theorem~3]{tandon2016toward}, which reduces to 
    \begin{align*}
        \FsII \leq \left\{\begin{array}{ll}
            \!\!\!(d-\ell)^2\alpha/d, & \!1\hspace{-1pt}\leq\hspace{-1pt} \ell \leq \min(n\!-\!d,d/2),\\
            \!\!\!(d-\ell)(d-1)\alpha/d,  &\! \min(n\!-\!d,d/2)\hspace{-1pt}<\hspace{-1pt}\ell \leq d, 
        \end{array}
        \right.
    \end{align*}
    for $k=d$. For the achievable tuples in Theorem~\ref{thm:achv_II} we have 
    \begin{align*}
        \frac{(d-\ell)^2}{d}\alpha^{(m)} &= \frac{(d-\ell)^2}{d} \binom{d}{m}
        \nonumber\\    
        &= (d-\ell)\frac{d-\ell}{m} \binom{d-1}{m-1}
        \nonumber\\ 
        &\geq (d-\ell)\frac{d-\ell +1}{m+1} \binom{d-\ell-1}{m-1}
        \nonumber\\ 
        &= m\binom{d-\ell+1}{m+1}=\FsII^{(m)},
    \end{align*}
    where the inequality holds since $\frac{d-\ell}{m} \leq \frac{d-\ell+1}{m+1}$ for ${m\leq d-\ell}$ and the fact that $\binom{x}{m-1}$ is an increasing function of $x$. Hence, our codes satisfy the upper bound for all values of $m$. Note that the second bound for ${\min(n-d,d/2)<\ell \leq d}$ is looser compared to the first one, and hence, it is clearly satisfied by the proposed codes. 
    \item The secrecy capacity of a $(4,3,3,1)$-DSS is characterized by~\cite[Theorem~4]{tandon2016toward} as 
    \[\FsI = \min(\min(\alpha,2\beta) + \min(\alpha,\beta), (\alpha+6\beta)/3),
    \]
    which has three (normalized) extreme points $(1,1/3)$, $(3/5,2/5)$, and $(1/2, 1/2)$. The exact set of tuples can be achieved using the secure determinant codes with parameters given in Theorem~\ref{thm:achv_I} for $m=1$, $m=2$, and $m=3$, respectively. Moreover, the Type-II secrecy capacity of a $(4,3,3,1)$-DSS is given in~\cite[Theorem~4]{tandon2016toward} by $\FsII=\min(\alpha, 3\beta)$, which has a single normalized extreme point $(1, 1/3)$. This coincides with the achievable tuple of Theorem~\ref{thm:achv_II} for $m=1$. Hence, the proposed secure determinant codes are optimum for a $(4,3,3,1)$ system. 
    \item A class of Type-II secure codes for a DSS with parameters ${(n=d+1, k=d, d, \ell)}$ is proposed in~\cite[Theorem~2]{shao2017tradeoff}, that satisfy 
    \begin{align*}
        (\alpha_t, \beta_t, F_{\mathsf{s,II},t}) = \left(\frac{1}{t\!-\!1}\binom{n\!-\!1}{t\!-\!1}, \frac{1}{d}\binom{n\!-\!1}{t\!-\!1}, \binom{n\!-\!\ell}{t} \!\!\right)\!,\!
    \end{align*}
    where $t\in \{2,\dots, n-\ell\}$. 
    It is worth noting that these parameters exactly match those in Theorem~\ref{thm:achv_II}. More precisely, it is straightforward to show that
    \[(\alpha_t, \beta_t, F_{\mathsf{s,II},t}) = \frac{1}{t-1} \left(\alpha^{(t-1)}, \beta^{(t-1)}, F_{\mathsf{s,II}}^{(t-1)}\right).\] However, while the construction in~\cite{shao2017tradeoff} is limited to $n=d+1$, the proposed construction in this work can be applied to any number of nodes in the system. 
    \item The Type-II secrecy trade-off of a $(7,6,6,1)$-DSS is characterized in~\cite[Theorem~4]{shao2017tradeoff}, which is shown to have two (normalized) extreme points, namely, $(2/5, 1/15)$ and $(3/8,1/8)$. It is easy to verify that the codes introduced in Theorem~\ref{thm:achv_II} achieve ${(\alpha^{(1)},\beta^{(1)}, \FsII^{(1)}) = (6,1,15)}$ and ${(\alpha^{(2)},\beta^{(2)}, \FsII^{(2)}) = (15,5,40)}$, which lead to the same normalized pairs. Note that, however, Theorem~\ref{thm:achv_II} provides codes with the same parameters for an arbitrary $n$, and it is not limited to $n=7$. 
    \item It is reported in~\cite{ye2019secure} that the optimum trade-off of an $(n,k=d,d,\ell)$ system has a single Pareto point (i.e., $t=1$ in Property~\ref{cor:num_pareto}) if and only if $\ell \geq \Bigl\lceil{\frac{d-1}{4}}\Bigr\rceil$. This is a special case of Property~\ref{cor:num_pareto} by rewriting the condition in~\eqref{eq:con:pareto} as $2\geq \frac{\sqrt{1+4\ell (d+1)}-1}{2\ell}$. 
\end{itemize}

\section{A Brief Review of Determinant Codes}
\label{sec:det}
We use the \emph{non-secure} determinant codes \cite{elyasi2019determinant} as the main building block to construct secure exact-repair regenerating codes. Before presenting the proposed construction for Type-I and Type-II security, we start with a brief review of determinant codes. 

Consider a DSS with system parameters $(n,k=d,d)$. A collection of $d$ different determinant codes can be constructed for this system. They are labeled by a \emph{mode} parameter ${m \in [d]}$. The parameters of the determinant code with mode ${m\in [d]}$ are given in~\cite{elyasi2019determinant} by
\begin{align}
    \label{param}
    \left(F^{(m)},\alpha^{(m)},\beta^{(m)}\right)=\left({m{\binom{d+1}{m+1}},\binom{d}{m}},{\binom{d-1}{m-1}}\right).
\end{align}
Determinant codes operate at different corner points of the trade-off curve between $\alpha$ and $\beta$, by varying the mode from the MBR point with $m=1$ to the MSR point with $m=d$. Next, we present the code construction of determinant codes at a given mode $m$. In what follows, we fix $m$ and drop the superscript for ease of notation.

Our goal is to store a total of $F = m\binom{d+1}{m+1}$ source symbols from some $\mathbb{F}_q$ in the DSS. To this end, we need to construct the \emph{message matrix} $\bM$, which has $d$ rows and $\alpha=\binom{d}{m}$ columns. The rows of $\bM$ are labeled by $x \in [d]$, while the columns of $\bM$ are labeled by subsets $\cI\subseteq [d]$ of size $|\cI|=m$, sorted in lexicographical order.

\begin{defi}\label{def:types}
    For fixed parameters $d$ and $m\in [d]$, we define the types (sets)  $\cV$, $\cW$, and $\cW$ as
    \begin{align}\label{eq:types}
        \left\{
        \begin{array}{ll}
            \cV=\{(x,\cI):  x\in \cI \!\subseteq\! [d] ,\: |\cI|\!=\!m \}, \\
            \cW=\{(x,\cI): x\in [d] \!\setminus\! \cI,\:  x\!<\!\max \cI,\: \cI\!\subseteq\! [d] ,\: |\cI|\!=\!m \},\\
            \cP=\{(x,\cI):x\in [d]\!\setminus\! \cI,\: x\!>\!\max \cI,\: \cI\!\subseteq\! [d],\: |\cI|\!=\!m\}.
        \end{array}
    \right.
    \end{align}
    Moreover, for a matrix $\bA$ of size $d\times \binom{d}{m}$, we use $\cV(\bA)$, $\cW(\bA)$, and $\cP(\bA)$ to refer to the collection of entries of $\bA$ at positions belong to $\cV$, $\cW$, and $\cP$, respectively.   
\end{defi}
The following remark specifies the size of the sets defined above. 
\begin{remark}\label{rem:size:type}
    Note that there are $\binom{d}{m}$ choices for $\cI$, and if position $(x,\cI)$ is $\cV$-type as defined in~\eqref{eq:types}, then $x$ can be any element of $\cI$. Hence, $|\cV|=m\binom{d}{m}$. Each $\cW$-type pair $(x,\cI)$ corresponds to a set $\cJ = \cI \cup \{x\} \subseteq [d]$ with $|\cJ|=m+1$, where $x$ can be any element of $\cJ$ except the maximum one. Therefore, we have $|\cW|=m\binom{d}{m+1}$. Finally, each $(x,\cI)\in \cP$ is corresponding to a set $\cJ = \cI \cup \{x\} \subseteq [d]$ with $|\cJ|=m+1$, where $x= \max \cJ$.  Thus, we get $|\cP(\bM)|=\binom{d}{m+1}$. Note that we have  $|\cV| + |\cW| + |\cP|  =  m\binom{d+1}{m+1} + \binom{d}{m+1} = d\binom{d}{m} = d\alpha$, which is the number of entries in a matrix with $d$ rows and $\binom{d}{m}$ columns. 
\end{remark}
 
Next, we determine the entries of the message matrix $\bM$. We fill all the $\cV$-type and $\cW$-type positions of $\bM$ with the information symbols. Therefore, from Remark~\ref{rem:size:type} we have  $|\cV(\bM)|+|\cW(\bM)| = m\binom{d}{m} + m \binom{d}{m+1} = m\binom{d+1}{m+1} = F$.

Each $\cP$-type entry of the message matrix at position $(x,\cI)$ with $x>\max \cI$ will be filled by a \emph{parity symbol}, which is given by
\begin{align*}
    \mat{\bM}{x,\cI} \hspace{-1pt}=\hspace{-1pt} (-1)^{m}\sum_{y\in \cI}   (-1)^{\ind{\cI}{y}} \mat{\bM}{y,\cI\cup\{x\} \hspace{-1pt}\setminus\hspace{-1pt}\{y\}}.
\end{align*}
It is worth noting that  ${y\leq \max \cI < x = \max \cI\cup\{x\} \setminus\{y\}}$ and ${y\notin \cI\cup\{x\} \setminus\{y\}}$, and hence $\mat{\bM}{y,\cI\cup\{x\} \setminus\{y\}}$ is a {$\cW$-type} entry. In other words, for every $\cJ\subseteq [d]$ with ${|\cJ|=m+1}$, the matrix entries in ${\{ \mat{\bM}{y,\cJ \setminus\{y\}}: y\in \cJ\}}$ satisfy a parity equation, given by 
\begin{align} 
    \sum_{y\in\cJ}  (-1)^{\text{ind}_{\cJ}(y)}\mat{\bM}{y,\cJ \setminus\{y\}} = 0.
    \label{eq:parity}
\end{align}
It is worth noting that such a group includes $m$ matrix entries with type $\cW$ and a single $\cP$-type element. We refer to this set as \emph{parity group $\cJ$} in the rest of this paper. It is worth noting that the parity groups are disjoint, and we have exactly one $\cP$-type entry per parity group. We may use $v_{x,\cI}$ to refer to an entry of type $\cV$ at position $(x,\cI)$. Similarly, the entry at position $(x,\cI)$ from type $\cW$ or $\cP$ will be referred to as $w_{x,\cJ}$, where $\cJ=\cI\cup \{x\}$.     

Next, we select an $n\times d$  Vandermonde matrix\footnote{For a general determinant code, the encoder matrix $\mathbf{\Psi}$ can be any $n\times d$ matrix whose all $d\times d$ sub-matrices are  full-rank. However, we set it to be a Vandermonde matrix here, which is more convenient for the secrecy constraints.} $\mathbf{\Psi}$ to be used as the \emph{encoder matrix}. The entries of $\mathbf{\Psi}$ are drawn from a finite field $\mathbb{F}_q$, that includes at least $n$  distinct non-zero entries. Hence, we have $\mat{\mathbf{\Psi}}{i,j}= \psi_j^{i}$, where $\psi_1,\dots, \psi_d$ are distinct elements of $\mathbb{F}_q$.  Finally, the \emph{determinant code matrix} $\bC$ is constructed by multiplying the message matrix $\bM$ by the encoder matrix $\mathbf{\Psi}$, that~is, 
\begin{align} \label{construct}
    \bC_{n\times\alpha} = 
    \begin{bmatrix}
        \cN_1 \\
        \vdots \\
        \cN_n 
    \end{bmatrix} 
    = \mathbf{\Psi}_{n\times d} \cdot \bM_{d\times\alpha}. 
\end{align}
The content of node $i$ is denoted by the row vector $\cN_i$, which is the $i$th row of $\bC$ and consists of $\alpha$ symbols\footnote{It may appear at the first glance that the determinant codes are similar to the product-matrix (PM) codes~\cite{rashmi2011optimal, shah2011information}. First note that while PM code construction is limited to the MBR and MSR points, the determinant codes are capable of operating at the intermediate points on the $\alpha-\beta$ trade-off.  Moreover, while the MBR-PM code is equivalent to a determinant code at mode $m=1$, the MSR-PM code is fundamentally different from a determinant code at mode $m=d$, and the two codes cannot be converted to each other by a change of basis. We refer to~\cite{elyasi2020cascade} for further discussions.}. 

The data recovery property of the code is an immediate consequence of the MDS property of the encoder matrix. More specifically, by accessing the content of any subset of $|\cK|=k=d$, and stacking the corresponding rows of $\bC$, we can recover the matrix ${ \mat{\mathbf{C}}{\cK,:} = \mat{\mathbf{\Psi}}{\cK,:} \cdot \bM}$, where  $\mat{\mathbf{C}}{\cK,:}$ and $\mat{\mathbf{\Psi}}{\cK,:}$ are sub-matrices of $\mathbf{C}$ and $\mathbf{\Psi}$ obtained from the rows whose labels belong to $\cK$. Note that $\mat{\mathbf{C}}{\cK,:}$  is a $d\times d$ Vandermonde matrix, and so it is full-rank. Therefore,  we can recover the message matrix $\bM$ by multiplying $\mat{\mathbf{C}}{\cK,:}$ by $\mat{\mathbf{\Psi}^{-1}}{\cK,:}$, i.e., ${\bM = \mat{\mathbf{\Psi}^{-1}}{\cK,:} \cdot \mat{\mathbf{C}}{\cK,:}}$. 

Upon failure of node $f\in[n]$ and selection of a set of helper nodes  $\cH\subseteq[n]\backslash\{f\}$ with $|\cH|=d$, the repair data from node $h\in \cH$ for node $f$, denoted by $\R{h}{f}$, is given by 
\begin{align}
    \R{h}{f}=\cN_h \cdot\bxi^{f} = \mat{\mathbf{\Psi}}{h,:}\cdot\bM\cdot\bxi^{f}.
    \label{eq:repair-data}
\end{align}
Here, $\cN_h$ is the content of node $h$, $\mat{\mathbf{\Psi}}{h,:}$ is the $h$th row of $\mathbf{\Psi}$, and $\bxi^{f}$ is a $\binom{d}{m}\times\binom{d}{m-1}$ \emph{repair encoder} matrix, whose rows and columns are indexed by subsets $\cI, \cJ \subseteq [d]$ with $|\cI|=m$ and $|\cJ| = m-1$. The entry of $\bxi^{f}$ at position $(\cI,\cJ)$ is given by
\begin{align}
    \mat{\bxi^{f}}{\cI,\cJ}\hspace{-1pt}=\hspace{-1pt} \left\{ 
    \begin{array}{ll} 
        \!(-1)^{\text{ind}_\cI(x)}\mat{\mathbf{\Psi}}{f,x} & \mbox{if\ } \cJ\hspace{-1pt}\cup\hspace{-1pt} \{x\}\hspace{-1pt}=\hspace{-1pt} \cI,\\ 
        \!0 & \mbox{otherwise.}
    \end{array}\right.
    \label{eq:xi-def}
\end{align}
Even though $\R{h}{f}$ is a vector of length $\binom{d}{m-1}$, it is shown in~\cite[Proposition~1]{elyasi2019determinant} that the rank of matrix $\bxi^{f}$ is ${\beta=\binom{d-1}{m-1}}$, and hence the vector $\R{h}{f}$ can be sent from helper node $h$ to the failed node $f$ by communicating  $\beta = \binom{d-1}{m-1}$ entries of $\R{h}{f}$, and hence the per-node repair bandwidth constraint is satisfied.  Also, it is shown in~\cite[Proposition~2]{elyasi2019determinant} that the content of node $f$, i.e., $\cN_f =\mat{\mathbf{\Psi}}{f,:}\cdot  \bM$, can be retrieved from $\{\R{h}{f}: h\in \cH\}$ for any $\cH\subseteq [d]\setminus\{f\}$ with $|\cH|=d$. More precisely, the $\cI$th element of node $f$ can be recovered from 
\begin{align}
    [\mat{\mathbf{\Psi}}{f,:}\cdot\bM]_\cI = \sum_{x\in \cI} (-1)^{\text{ind}_\cI(x)} \mat{\bR^{f}}{x,\cI\setminus\{x\}},
\end{align}
where $\bR^{f}$ is a $d\times \binom{d}{m-1}$ matrix defined as
\begin{align}\label{eq:rep-2}
    \bR^{f} := \mat{\mathbf{\Psi}^{-1}}{\cH,:} \begin{bmatrix}
    \R{h_1}{f} \\ \R{h_2}{f} \\ \vdots \\ \R{h_d}{f}
    \end{bmatrix}.
\end{align}
Note that the matrix  $\mat{\mathbf{\Psi}}{\cH,:}$ is a $d\times d$ sub-matrix of $\mathbf{\Psi}$ which is a full-rank Vandermonde matrix, and hence, is an invertible matrix. Finally, the latter matrix in~\eqref{eq:rep-2} can be formed at the failed node $f$, by stacking all the repair data received from the helper nodes. 

\section{Type-I Secure Determinant Codes}
\label{sec:typeI}
\subsection{Code Construction for Type-I Security}
\label{sec:codeConstruct_I}
In this section, we present the construction of Type-I secure determinant codes for Theorem~\ref{thm:achv_I}. Consider a DSS with system parameters  ${(n,k=d,d,\ell)}$ and a given mode $m\in[d]$.  The goal of Theorem~\ref{thm:achv_I} is to securely store $\FsI=\FsI^{(m)}$ symbols in a determinant code with parameters $(\alpha,\beta)=\left(\alpha^{(m)}, \beta^{(m)}\right)$, as given by~\eqref{eq:thm_achv_I}. The construction of the secure code is similar to that of the (non-secure) determinant code, in which the secure information symbols as well as a set of randomly generated symbols are stored. Let $\cS$ denote the set of secure file symbols, where $|\cS|=\FsI$. Moreover, let  $\key$ be the set of 
\begin{align}
    \label{numberOfKeys_I}
    |\key| = \ell \alpha - \binom{\ell}{m+1} = \ell\binom{d}{m} - \binom{\ell}{m+1} 
\end{align}
random symbols, drawn uniformly and independently (from each other and from the secure file symbols) from $\mathbb{F}_q$. The symbols in $\key$ play the role of random keys in the code construction. Note that 
\begin{align*}
    &|\cS\cup \cQ| 
    \nonumber\\
    &= |\cS|+|\cQ| = \FsI + |\key| 
    \nonumber\\
    &= (d\!-\!\ell) \binom{d}{m} \!-\! \binom{d}{m\!+\!1} \!+\! \binom{\ell}{m\!+\!1} \!+\! \ell\binom{d}{m} \!-\! \binom{\ell}{m\!+\!1}
    \nonumber\\
    &= d \binom{d}{m} - \binom{d}{m+1} = F^{(m)}.
\end{align*}

As mentioned before, the construction of Type-I secure determinant codes is similar in spirit to that of non-secure determinant codes presented in Section~\ref{sec:det}, in which the information symbols comprise the union of the secure file symbols and the random keys. However, a key ingredient in the proposed construction is to opportunistically choose the position of the secure symbols and the random keys in the message matrix $\MsI$, in order to guarantee security against a Type-I eavesdropper. To this end, we fill all the entries in the top $\ell$ rows of $\MsI$ using the key symbols. More precisely, all the entries  $(x,\cI)\in \cV(\MsI) \cup \cW(\MsI)$ with $x\in  [\ell]$ will be by an element from $\key$. Similarly, each entry $(x,\cI) \in \cV(\MsI) \cup \cW(\MsI)$ with $x\in [\ell+1:d]$ will be filled by a secure file symbol from~$\cS$. The parity entries in $\cP(\MsI)$ will be generated according to the parity equation~\eqref{eq:parity}.     

Note that there are $\ell\alpha$ symbols in the top $\ell$ rows. However, if $\ell\geq m+1$, one parity symbol is needed to be introduced for each group of $m+1$ of such rows. Thus,   $\binom{\ell}{m+1}$ of the symbols in the top $\ell$ rows are parity symbols. This leads to 
\begin{align*}
    &|\{ (x,\cI)\in \cV(\MsI) \cup \cW(\MsI): x\in[\ell]\} | 
    \nonumber\\
    &= |\{(x,\cI): x\in[\ell]\} | - |\{(x,\cI)\in \cP(\MsI): x\in [\ell] \}|
    \nonumber\\
    &= 
    \ell \alpha - \binom{\ell}{m+1},
\end{align*}
which is consistent with the number of random keys as given in~\eqref{numberOfKeys_I}. 

The encoder matrix $\mathbf{\Psi}$ is an $n\times d$ matrix with entries from $\mathbb{F}_q$ that satisfies two properties: 
\begin{enumerate}[label= {(C\arabic*)}, ref={(C\arabic*)}]
    \item  any $d\times d$ sub-matrix of $\mathbf{\Psi}$ is full rank;  \label{cond:psi:1}
    \item  and any $\ell\times\ell$ sub-matrix of $\mat{\mathbf{\Psi}}{:,[\ell]}$ is full rank. \label{cond:psi:2}
\end{enumerate}
It is convenient to choose a Vandermonde matrix for $\mathbf{\Psi}$, which satisfies both  properties~\ref{cond:psi:1} and~\ref{cond:psi:2}. Let $\mathbf{\Psi}$ be a Vandermonde matrix generated by distinct (non-zero) elements  $x_1,x_2,\dots, x_n \in \mathbb{F}_q$, i.e., $\mathbf{\Psi}(i,j) = x_i^{j-1}$. Then, Condition~\ref{cond:psi:1} is an immediate property of the Vandermonde structure. Moreover, for an arbitrary set of $\ell$ rows ${\cL=\{i_1,i_2,\dots, i_\ell\} \subset [n]}$, the matrix $\mathbf{\Psi}(\cL, [\ell])$ is also a  Vandermonde matrix generated by  ${x_{i_1}, x_{i_2}, \dots, x_{i_\ell} \in \mathbb{F}_q}$, and we have ${\mathsf{det}(\mathbf{\Psi}) = \prod_{1\leq s <t \leq \ell} (x_{i_t}-x_{i_s})\neq 0}$. Therefore, Condition~\ref{cond:psi:2} is also satisfied.  Note that, in order to construct a Vandermonde matrix $\mathbf{\Psi}$ of size $n\times d$, it is required that $q=|\mathbb{F}_q| >n$. This is the \emph{only} constraint on the field size imposed by the code construction.  

\begin{remark}
    The proposed codes require a fairly small field size since it is only constrained by the existence of an $n \times d$ encoder matrix satisfying Conditions~\ref{cond:psi:1}~and~\ref{cond:psi:2}. This is guaranteed by Vandermonde matrices constructed by $n$ distinct elements of $\mathbb{F}_q$. Hence, $q$ is restricted to be a prime power (for the existence of a finite field of size $q$) and $q>n$ (for the existence of $n$ distinct and non-zero elements in the field). To be more precise, we can use Bertrand's postulate that guarantees the existence of a prime number $q$ satisfying $n<q<2n$ (for~${n>1}$), and conclude that $q<2n$. This shows that $q$ does not need to grow faster than $n$, and we have $q = \Theta(n)$. 
\end{remark}

Finally, once matrices $\MsI$ and $\mathbf{\Psi}$ are generated, the content of nodes will be determined by the rows of $\CsI = \mathbf{\Psi}\cdot\MsI$, similar to~\eqref{construct}. 

\subsection{An Illustrative Example for Type-I Security}
\label{sec:codeEx_I}
In this section, we present an example of the code construction for a Type-I secure determinant code. Consider a $(n,k,d,\ell) = (n,6,6,2)$-DSS operating at mode $m=2$. For illustrative purposes, we first present the code construction of the non-secure determinant code whose parameters are $(F, \alpha, \beta) = (70, 15, 5)$, as given by~\eqref{param}. Figure~\ref{fig:ex_nonSecure} depicts the corresponding message matrix $\bM$, where rows are indexed by integers from $[d]=[6]$ and columns are subsets of size $m=2$ with entries from $[d]=[6]$.  The  $\cV$-type symbols are shown in solid gray boxes, while symbols of type $\cW$ or $\cP$  are depicted in dotted boxes with different background colors, where  each background color indicates one parity group (see~\eqref{eq:parity}). Recall that since $d=6$ and $m=2$, each parity group corresponds to a subset $\cJ\in[6]$ with $|\cJ|=m+1=3$. Note that the (non-secure) determinant code for $n$ storage nodes can be obtained by multiplying $\bM$ by an encoder matrix $\mathbf{\Psi}_{n\times 6}$, as given by~\eqref{construct}.

Now, we shift our attention to the construction of \mbox{Type-I} determinant code where the message should be secured against Type-I eavesdroppers, who can access the coded content of up to $\ell = 2$ nodes. The code parameters are ${(\FsI,\alpha,\beta) = (40, 15, 5)}$, as given in~\eqref{eq:thm_achv_I}. Moreover,~\eqref{numberOfKeys_I} implies that we need to use $|\key| = 30$ random keys. Note that the storage capacity of the system reduces from $F = 70$ to $\FsI = 40$ in order to guarantee security against Type-I eavesdroppers. More precisely, even  though we still use $70$~symbols to fill the entries of matrix $\MsI$, only $\FsI=40$ of them are secure information symbols, and the remaining $30$~symbols are randomly generated keys. 

Let us denote of secure symbols by $\cS = \{u_1, u_2, \cdots, u_{40}\}$, and label the random keys by $\key = \{r_1, r_2, \cdots, r_{30}\}$. Figure~\ref{fig:ex_secure_I} depicts the corresponding message matrix $\MsI$. The random keys (in pink boxes) are placed in the top  $\ell=2$ rows, and the secure symbols (in blue boxes) are placed in the bottom  ${d-\ell = 4}$ rows. Parity symbols (in green dotted boxes) are generated according to the parity equations in~\eqref{eq:parity}. For instance, for the parity group $\cJ=\{1,3,4\}$ the  entries $\mat{\MsI}{4,\{1,3\}}$, $\mat{\MsI}{3,\{1,4\}}$ and $\mat{\MsI}{1,\{3,4\}}$ should satisfy the parity equation~\eqref{eq:parity}, i.e.,  
\begin{align}
    0 &= \sum_{y\in \{1,3,4\}} 
    (-1)^{\text{ind}_{\{1,3,4\}}(y)} \mat{\MsI}{y,\{1,3,4\} \setminus y}
    \nonumber\\
    & = 
    (-1)^1 \mat{\MsI}{1,\{3,4\}} + (-1)^2 \mat{\MsI}{3,\{1,4\}} 
    \nonumber\\
    & \phantom{=} + (-1)^3 \mat{\MsI}{4,\{1,3\}}.
    \nonumber
\end{align}
This determines the parity symbol $w_{4,\{1,3,4\}}$, which will be placed in $\mat{\MsI}{4,\{1,3\}}$ as
\begin{align*}
     \mat{\MsI}{4,\{1,3\}}
     &\hspace{-2pt}= \hspace{-1pt}(-1)^1 \mat{\MsI}{1,\{3,4\}} \hspace{-1pt}+\hspace{-1pt} (-1)^2 \mat{\MsI}{3,\{1,4\}} 
    \nonumber\\
    & = -r_{10} + u_2.
\end{align*}
It is worth noting that since $\ell<m+1$, we have $\binom{\ell}{m+1}=0$, and there is no parity symbol within the top $\ell$ rows of $\MsI$.

Finally, the Type-I secure determinant code for $n$ storage nodes can be obtained by multiplying $\bM$ by a Vandermonde matrix $\mathbf{\Psi}_{n\times 6}$ with $\mathbf{\Psi}=\psi_j^i$, as given by \eqref{construct}. Each node stores $\alpha=15$ coded symbols. For instance, the first and second coded symbols stored in node $i$ are given by 
\begin{align*}
    \cN_{i}(\{1,2\})=&\psi_1^i r_1 +\psi_2^i r_{16} +\psi_3^i (r_{17}-r_6) + \psi_4^i (r_{18}-r_7), \\
    & +\psi_5^i (r_{19}-r_8) +\psi_6^i (r_{20}-r_9)\\ 
    \cN_{i}(\{1,3\})=&\psi_1^i r_2 +\psi_2^i r_{17} +\psi_3^i u_1 + \psi_4^i (u_2-r_{10}) \\ &+\psi_5^i (u_3-r_{11})+\psi_6^i (u_4-r_{12}).
\end{align*} 

It should be noted that the construction of Type-I secure determinant code inherits the data recovery and node repair properties from non-secure determinant code construction. Therefore, all secure symbols as well as random keys can be reconstructed from the contents of any set of $k = 6$ nodes. Moreover, any failed node can be repaired by receiving repair data from any set of $d = 6$ nodes, and downloading $\beta = 5$ repair symbols from each helper node. In the rest of this section, we will prove that the proposed code is secure against any Type-I eavesdropper. 

\subsection{Proof of Theorem \ref{thm:achv_I}}
\label{sec:achvProof_I}

The proposed code construction is a secure version of the determinant code, which is secure against Type-I eavesdroppers. It is shown in Section~\ref{sec:codeConstruct_I} that the code parameters match the values given in~\eqref{eq:thm_achv_I}. The parameters of the codes. Due to its construction, it is evident that it maintains the \textit{Data Recovery} property due to \cite[Proposition~1]{elyasi2016determinant}. It also preserves the \textit{Node Repair} property due to \cite[Proposition~1]{elyasi2019determinant}. It remains to prove that the secure determinant code proposed in Section~\ref{sec:codeConstruct_I} satisfies the Type-I security constraint in~\eqref{req_mut_I}. To this end, we introduce two key lemmas essential for the proof of Type-I security property. We refer to Appendices~\ref{prf:lm:secProof_I_lm2} and~\ref{prf:lm:secProof_I_lm1} for the proof of Lemmas~\ref{lm:secProof_I_lm2} and~\ref{lm:secProof_I_lm1}, respectively. 

\begin{lm}
    \label{lm:secProof_I_lm2}
    For every $\cL \subseteq [n]$ with $|\cL|\leq\ell$, the entropy of the eavesdropper's observation  $\cEI(\cL)$, in an $(n,d,d)$ determinant code of mode $m$ is upper bounded by the number of key symbols,  i.e.,
    \begin{align}
        H(\cEI(\cL))\leq |\key| = \ell \binom{d}{m} - \binom{\ell}{m+1}, \quad 
        \forall \cL  \subseteq  [n],\:  |\cL| \leq \ell.
    \end{align}
\end{lm}
 
\begin{lm}
    \label{lm:secProof_I_lm1}
    For a determinant code generated according to the construction of Section~\ref{sec:codeConstruct_I},  the set of random keys can be fully recovered given the secure message $\cS$ and the eavesdropper's observation $\cEI(\cL)$, for every $\cL \subseteq [n]$ with $|\cL|= \ell$, i.e.,
    \begin{align}\label{eq:I:Q-from-E&S}
        H(\key | \cEI(\cL),\cS)=0, \qquad 
        \forall \cL \subseteq [n] \text{ and } |\cL|=\ell.
    \end{align}
\end{lm}  

Now, we are ready to prove that the proposed coded construction satisfies the Type-I security constraint in~\eqref{req_mut_I}. For any $\cL \subseteq [n]$ and $|\cL|\leq\ell$, we have
\begin{align}
    I(\cS;\cEI(\cL)) & = H(\cEI(\cL))-H(\cEI(\cL)|\cS) 
    \nonumber\\
    & \stackrel{\rm{(a)}}{\leq} |\key|-H(\cEI(\cL)|\cS) 
    \nonumber\\
    & \stackrel{\rm{(b)}}{=} |\key| - H(\cEI(\cL)|\cS) + H(\cEI(\cL)|S,\key) 
    \nonumber\\
    & = |\key|-I(\cEI(\cL); \key |\cS) 
    \nonumber\\
    & = |\key| - H(\key|\cS) + H(\key|\cEI(\cL),\cS) 
    \nonumber\\
    & \stackrel{\rm{(c)}}{=} |\key|-H(\key|\cS) 
    \nonumber\\
    & \stackrel{\rm{(d)}}{=} |\key|-|\key| = 0,
    \nonumber
\end{align}
where {\rm (a)} follows from Lemma~\ref{lm:secProof_I_lm2},  in~{\rm (b)} we used the fact that the node contents are all deterministic functions of the secure message $\cS$ and the random keys $\key$ and hence $H(\cEI(\cL)|S,\cK)=0$, {\rm (c)} follows from Lemma~\ref{lm:secProof_I_lm1}, and {\rm (d)} holds since the random keys are independent of the secure message. This completes the proof of Theorem~\ref{thm:achv_I}. 
\hfill $\square$

\onecolumn
\begin{sidewaysfigure}
    \includegraphics[width=\columnwidth]{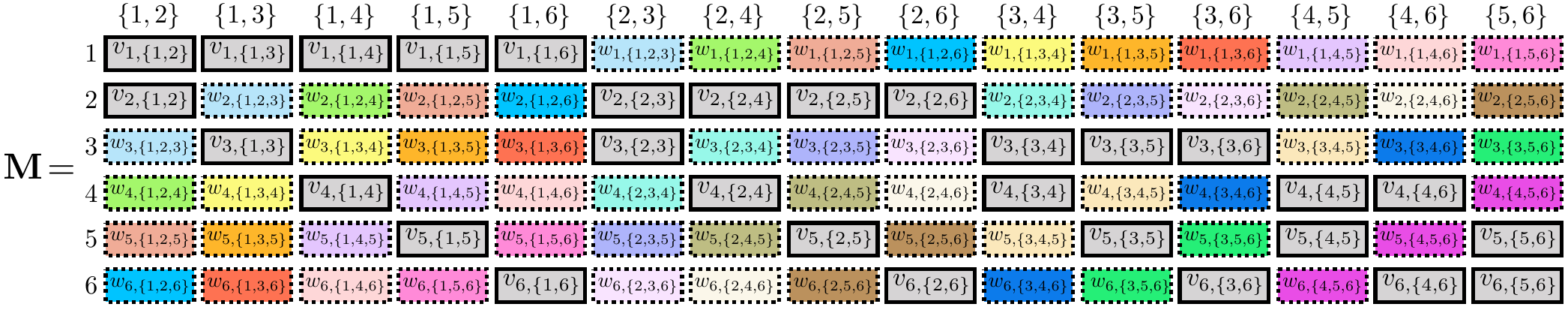}
    \caption{Message matrix $\bM$ (with row and column labels) for a non-secure determinant code with parameters $(n,k,d) = (n,6,6)$, and mode $m = 2$. Symbols from set $\cV$ are in solid gray boxes, while symbols from set $\cW$ are in dotted boxes and colored in different colors such that each set of symbols that satisfies \eqref{eq:parity} is colored with the same color.}
    \label{fig:ex_nonSecure}
    \vspace{1cm}
    \hrule
    \vspace{1cm}
    \includegraphics[width=\columnwidth]{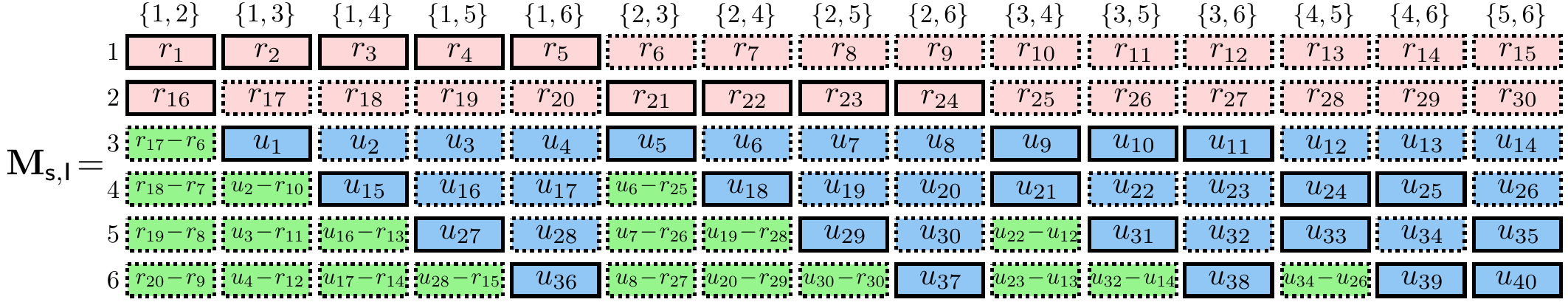}
    \caption{Message matrix $\MsI$ (with row and column labels) for a Type-I secure determinant code with parameters $(n,k,d,\ell) = (n,6,6,2)$, and mode $m = 2$. 
    Random keys are denoted by $r$ (in pink boxes), while secure symbols are denoted by $u$ (in blue boxes). The parity symbols  (in green dotted boxes) are generated such that the parity equations in \eqref{eq:parity} are satisfied.}
    \label{fig:ex_secure_I}
\end{sidewaysfigure}
\twocolumn

\subsection{The Secrecy Capacity of Type-I Secure Determinant Codes}
\label{sec:converseProof_I}
In this subsection, we present the proof of Property \ref{thm:converse_I} and provide a tight upper bound on the maximum file size to guarantee Type-I security for determinant codes.

Consider an $(n,k=d,d,\ell)$ Type-I secure distributed storage system. Without loss of generality, assume $\cL=[\ell]$ and $\cEI(\cL) = \cN_\cL$. From the data recovery property, the secure file is recoverable from the contents of any $k=d$ nodes. In particular, the entire secure message can be retrieved from the content of the nodes in $\cL \cup [\ell+1:d]$. On the other hand, the data repair property implies the content of a failed node $j \in [\ell+1:d]$ can be repaired using the repair symbols downloaded from the node contents in $\cN_{[d+1]\setminus\{j\}}$. Let $\cS$ be the secure message stored in a determinant code of mode $m$, that maintains security against a Type-I eavesdropper with access to $\ell$ nodes. Then, we have
\begin{align}
   H(\cS) 
    &\stackrel{\rm{(a)}}{=} H(\cS) -I(\cS ; \cN_\cL)\nonumber\\
    & = 
    H(\cS \:|\: \cN_{[\ell]}) 
    \nonumber\\
    & \leq 
    H\left(\cS, \cN_{[\ell+1,d]} \:|\: \cN_{[\ell]}\right)
    \nonumber\\
    & = 
    H\left(\cN_{[\ell+1,d]} \:|\: \cN_{[\ell]}\right)
    + 
    H\left(\cS \:\middle|\: \cN_{[\ell]}, \cN_{[\ell+1,d]}\right)
    \nonumber\\
    & \stackrel{\rm{(b)}}{=} 
    H\left(\cN_{[\ell+1,d]} \:|\: \cN_{[\ell]}\right)
    \nonumber\\
    & \leq 
    H\Big(\cN_{[\ell+1,d]},\:   \bigcup_{j=\ell+1}^d 
    \bigcup_{i=j+1}^{d+1}\R{i}{j} \:\Big|\: \cN_{[\ell]}\Big)
    \nonumber\\
    & =  
    H\left(\bigcup_{j=\ell+1}^d \bigcup_{i=j+1}^{d+1} \R{i}{j} \:\middle|\: \cN_{[\ell]}\right) 
    \nonumber\\
    & \phantom{=}
    + H\left(\cN_{[\ell+1:d]} \:\middle|\: \cN_{[\ell]},\: \bigcup_{j=\ell+1}^d \bigcup_{i=j+1}^{d+1} \R{i}{j}\right)
    \nonumber\\
    & \leq  
    H\left(\bigcup_{j=\ell+1}^d \bigcup_{i=j+1}^{d+1} \R{i}{j}\right) 
    \nonumber\\
    & \phantom{=}
    + H\left(\cN_{[\ell+1:d]} \:|\: \cN_{[\ell]},\: \bigcup_{j=\ell+1}^d \bigcup_{i=j+1}^{d+1} \R{i}{j}\right)
    \nonumber\\
    & = 
    H\left(\bigcup_{i=\ell+2}^{d+1} \bigcup_{j=\ell+1}^{i-1} \R{i}{j}\right)
    \nonumber\\
    & \phantom{=}
    + \sum_{u=\ell+1}^d H\left(\cN_{u} \:\middle|\: \cN_{[u-1]},\: \bigcup_{j=\ell+1}^d \bigcup_{i=j+1}^{d+1} \R{i}{j}\right)
    \nonumber\\
    & \stackrel{\rm{(c)}}{\leq}  
    \sum\nolimits_{i=\ell+2}^{d+1} H\left( \bigcup_{j=\ell+1}^{i-1}  \R{i}{j}\right)
    \nonumber\\
    & \phantom{=}
    + \sum_{u=\ell+1}^d H\left(\cN_{u} \:\middle|\: \bigcup_{i=1}^{u-1} \R{i}{u},\:  \bigcup_{i=u+1}^{d+1} \R{i}{u}\right)
    \nonumber\\
    & \stackrel{\rm{(d)}}{=} 
    \sum_{i=\ell+2}^{d+1} \left[\binom{d}{m} \!-\! \binom{d-(i-\ell-1)}{m}\right]
    \nonumber
\end{align}
\begin{align}
    & \stackrel{\rm{(e)}}{=}  
    (d-\ell)\binom{d}{m} - \sum_{j=\ell}^{d-1} \binom{j}{m} 
    \nonumber\\
    & = 
    (d-\ell)\binom{d}{m} - \left[ \sum_{j=0}^{d-1} \binom{j}{m} - \sum_{j=0}^{\ell-1} \binom{j}{m} \right]
    \nonumber\\
    & \stackrel{\rm{(f)}}{=} 
    (d-\ell)\binom{d}{m} -\lb \binom{d}{m+1} - \binom{\ell}{m+1}\rb = \FsI^{(m)},
\end{align}
where {\rm (a)} follows from the Type-I security constraint in \eqref{req_mut_I}, {\rm (b)} follows from the data recovery property in \eqref{eq:recpvery}, and in~{\rm (c)} we have used the fact that $\R{i}{u}$ is a function of $\cN_i$. The first summation in  the RHS of~{\rm (c)} consists of $(d-\ell)$ terms, where the term corresponding to $i$ is the repair data that is sent from node $i$ in order to   simultaneously repair $i-\ell-1$ failed node. It is shown in~\cite[Theorem~2]{elyasi2019determinant} that in a determinant code of mode $m$ and for a set of simultaneously failed nodes $\cA$, the entropy of repair data sent from node $i$ satisfies \[H\Big(\bigcup\nolimits_{j\in \cA} \R{i}{j}\Big) = \beta_{|\cA|}^{(m)} = \binom{d}{m} - \binom{d-|\cA|}{m}.\] Moreover, each term in the second summation in the RHS of {\rm(c)} is zero due to  the node repair property in \eqref{eq:repair}, where $\cN_u$ can be retrieved from $\{\R{i}{u}:  i\in[d+1]\setminus\{u\}\}$. These together lead to {\rm(d)}. The equality in~\rm{(e)} is due to the change of variable $j=d-i+\ell+1$. Finally, we have used the binomial coefficient identity  $\sum_{i=0}^{a} \binom{i}{b}=\binom{a+1}{b+1}$ in~\rm{(f)}. This shows that the size of the secure message stored in a determinant code of mode $m$ cannot exceed $\FsI^{(m)}$, and hence the construction in Section~\ref{sec:codeEx_I} is optimum.  This completes the proof of Property~\ref{thm:converse_I}.  
\hfill$\square$

\section{Type-II Secure Determinant Codes}
\label{sec:typeII}
\subsection{Code Construction for Type-II Security}
\label{sec:codeConstruct_II}
In this subsection, we present the  construction for \mbox{Type-II} secure determinant codes and show that the proposed construction satisfies the parameters of Theorem~\ref{thm:achv_II}. Consider an $(n,k=d,d,\ell)$-DSS and a given mode $m\in[d]$. The goal is to securely store $\FsII^{(m)}$  symbols in a determinant code with parameters $\left(\alpha^{(m)}, \beta^{(m)}\right)$, which are given in~\eqref{eq:thm_achv_II}. Following the procedure for constructing Type-I secure determinant codes, let $\cS$ denote the set of secure file symbols where $|\cS|=\FsII$, and assume $\key$ denote the set of 
\begin{align}
    |\key|=m \binom{d+1}{m+1} -  m\binom{d-\ell+1}{m+1} 
    \label{numberOfKeys_II}
\end{align}
random keys, drawn independently (from each other and from the secure file symbols) and uniformly at random from some finite field $\mathbb{F}_q$. The construction of Type-II secure determinant codes is similar in spirit to the one of non-secure determinant codes presented in Section~\ref{sec:det}, where we fill the massage matrix $\MsII$ using the symbols in $\cS\cup \key$. It is easy to verify that 
\begin{align*}
    |\cS\cup \key|&= |\cS| + |\key| = \FsII+|\key|
    \nonumber\\ 
    &= m \binom{d\!-\!\ell\!+\!1}{m\!+\!1} +  \left(m \binom{d\!+\!1}{m\!+\!1} -  m\binom{d\!-\!\ell\!+\!1}{m\!+\!1} \right)
    \nonumber\\
    &= m\binom{d+1}{m+1}=F^{(m)}, 
\end{align*}
which is the number of symbols required to fill in the message matrix of a determinant code of mode $m$. However, a key ingredient in the proposed construction is to opportunistically choose the position of the secure symbols and the random keys in the message matrix $\MsII$ to guarantee security against Type-II eavesdroppers.

Consider a block decomposition of the message matrix $\MsII$, given by 
\begin{equation}
\label{eq:D-decompos}
\begin{tikzpicture}[
style1/.style={
  matrix of math nodes,
  every node/.append style={text width=#1,align=center,minimum height=5ex},
  nodes in empty cells,
  left delimiter=[,
  right delimiter=],
  },
style2/.style={
  matrix of math nodes,
  every node/.append style={text width=#1,align=center,minimum height=5ex},
  nodes in empty cells,
  left delimiter=\lbrace,
  right delimiter=\rbrace,
  }
]

\matrix[style1=0.5cm] (1mat)
{
  & & & & & &   \\
  & & & & & &   \\
  & & & & & &   \\
  & & & & & &   \\
  & & & & & &   \\
};
\draw[solid]
  (1mat-2-1.south west) -- (1mat-2-7.south east);
\draw[solid]
  (1mat-1-3.north east) -- (1mat-5-3.south east);
\node[font=\large] 
  at (1mat-1-2.south) {$\mathbf{A}$};
\node[font=\large] 
  at ([xshift=9pt]1mat-1-5.south) {$\mathbf{B}$};
\node[font=\large] 
  at (1mat-4-5.5) {$\mathbf{D}$};
\node[font=\large] 
  at (1mat-4-2) {$\mathbf{C}$};
\draw[decoration={brace,mirror,raise=5pt},decorate]
  (1mat-5-1.south west) -- 
  node[below=7pt] {$\binom{d}{m}-\binom{d-\ell}{m}$} 
  (1mat-5-3.south east);
\draw[decoration={brace,mirror,raise=5pt},decorate]
  (1mat-5-4.south west) -- 
  node[below=7pt] {$\binom{d-\ell}{m}$} 
  (1mat-5-7.south east);
\draw[decoration={brace,raise=12pt},decorate]
  (1mat-1-7.north east) -- 
  node[right=15pt] {$\ell$} 
  (1mat-2-7.south east);
\draw[decoration={brace,raise=12pt},decorate]
  (1mat-3-7.north east) -- 
  node[right=15pt] {$d-\ell$} 
  (1mat-5-7.south east);

\node at ([xshift=-27pt,yshift=-1.2pt]1mat.west) {$\MsII=$};

\end{tikzpicture}
\end{equation}
where the top part (submatrices $\bA$ and $\bB) $ has $\ell$ rows and the bottom part (submatrices $\bC$ and $\bD$) consists of $d-\ell$ rows. Similarly, the section on the left (submatrices $\bA$ and $\bC$) includes the first   $\binom{d}{m}-\binom{d-\ell}{m}$ columns of $\MsII$, while the section on the right (submatrices $\bB$ and $\bD$) consists of the last $\binom{d-\ell}{m}$ columns of $\MsII$. Then, any entry ${(x,\cI)\in \left(\cV(\MsII) \cup \cW(\MsII)\right) \cap \mathbf{D}}$ will be filled by the secure symbols. Similarly, an entry $(x,\cI)$ that lies in ${\left(\cV(\MsII) \cup \cW(\MsII)\right) \setminus \mathbf{D}}$ will be filled by a random key symbol. Finally, the parity symbols in $\cP(\MsII)$ will be filled according to the parity equation in~\eqref{eq:parity}. 

Similar to the Type-I code construction, we use an encoder matrix $\mathbf{\Psi}$ that satisfies conditions~\ref{cond:psi:1} and~\ref{cond:psi:2}. Lastly, having the message matrix $\MsII$ and $\mathbf{\Psi}$, the code will be generated as $\CsII = \mathbf{\Psi}\cdot\MsII$, and the $i$th row of matrix $\CsII$ will be stored in node $i$ of the DSS. 

\begin{remark}\label{rmk:bD}
    Consider an entry $(x,\cI)$ that lies in the submatrix~$\bD$. Clearly, we have $x\in [\ell+1:d]$. Moreover, since the columns of matrix $\MsII$ are labeled by subsets of $[d]$ of size $m$, sorted in lexicographical order, then we have $\cI\subseteq [\ell+1:d]$ for the column label $\cI$. Furthermore, for a parity symbol $(x,\cI)\in \cP(\MsII) \cap \bD$ with $x>\max \cI$ we have 
    \begin{align} 
        \bM(x, \cI) \!=\! 
        (-1)^m \sum_{y\in\cI} (-1)^{\text{ind}_{\cI}(y)} \: \bM(y,\cI\cup \{x\} \hspace{-1pt}\setminus \hspace{-1pt}\{y\}).
        \label{eq:parity-Ds}
    \end{align}
    Note that $y\in \cI$ and $\cI\subseteq [\ell+1:d]$ imply that $y\in [\ell+1:d]$. Moreover, $x\in[\ell+1:d]$ and $\cI\subseteq [\ell+1:d]$ imply that $\cI\cup \{x\} \setminus \{y\}$. Therefore, each entry $\bM(y,\cI\cup \{x\} \setminus \{y\})$ in~\eqref{eq:parity-Ds} is a secure information symbol, and hence the parity symbols in $\bD$ depend only on the information symbols, and not the random keys. In other words, the structure of submatrix $\bD$ is identical to that of the message matrix of a determinant code with parameter $\tilde{d}=d-\ell$. 
\end{remark}

\subsection{Illustrative Example for Type-II Security}
\label{sec:codeEx_II}
In this subsection, we present an example of the code construction for Type-II secure determinant codes. Consider a $(n,k,d,\ell) = (n,6,6,2)$ secure DSS operating at mode $m=2$. Note that this setting is the same as the one considered in the illustrative example for Type-I secure determinant codes in Section~\ref{sec:codeEx_I}. Hence, the code construction of the non-secure determinant code remains the same, as depicted by Fig~\ref{fig:ex_nonSecure}.

For a determinant code that is secure  against Type-II eavesdroppers who can access the incoming repair data from all nodes to up to $\ell = 2$ nodes, the parameters of the code are  $(\FsII,\alpha,\beta) = (20, 15, 5)$, as claimed in~\eqref{eq:thm_achv_II}. We also need $|\key| = 50$ random symbols, as determined in~\eqref{numberOfKeys_II}. It is worth noting that the Type-II security constraint is stronger than the Type-I security constraint, and hence the secrecy capacity of the system reduces from $\FsI = 40$ to $\FsII = 20$. Let the set of secure symbols be $\cS = \{u_1, u_2, \cdots, u_{20}\}$, and the set of random keys be $\key = \{r_1, r_2, \cdots, r_{50}\}$. Figure~\ref{fig:ex_secure_II} depicts the corresponding message matrix $\MsII$, where the symbols in $\cV(\MsII) \cup \cW(\MsII)$ are shown with solid boxes and the parity symbols in $\cP(\MsII)$ are identified with dashed boxes. The placement of the random keys (in solid pink boxes) and the secure information symbols (in solid blue boxes) follow the block matrix decomposition for the data matrix $\MsII$ in \eqref{eq:D-decompos}. More specifically, the matrix $\MsII$ is decomposed into four submatrices, where the secure symbols only appear in the bottom right block designated by row labels $\{3,4,5,6\}$, and column labels $\{\{3,4\}, \{3,5\}, \{3,6\}, \{4,5\}, \{4,6\}, \{5,6\}\}$. Furthermore, the parity symbols in pink dotted boxes are (only) functions of random keys, while the parity symbols in blue dashed boxes are (only) functions of secure symbols. For instance, matrix entries  $\mat{\MsII}{6,\{2,5\}}$ in submatrix $\bC$ can be found from \eqref{eq:parity} as 
\begin{align*}
    &\mat{\MsII}{6,\{2,5\}} 
    \nonumber\\
    &\quad\quad = (-1)^1 \mat{\MsII}{2,\{5,6\}} + (-1)^2 \mat{\MsII}{5,\{2,6\}} 
    \nonumber\\
    &\quad\quad = -r_{30} + r_{48}.
\end{align*}
Finally, the Type-II secure determinant code for $n$ storage nodes can be obtained by multiplying $\MsII$ by a Vandermonde matrix $\mathbf{\Psi}_{n\times 6}$, as given by \eqref{construct}. Similar to Type-I secure determinant codes, the construction of Type-II secure determinant codes inherits the data recovery and node repair properties from non-secure determinant code construction. Therefore, all secure symbols and random keys can be reconstructed from the contents of any set of $k =d= 6$ nodes. Moreover, any failed node can be repaired by downloading $\beta = 5$ repair symbols from each of $d = 6$ helper nodes. 

\begin{figure*}[!t]
    \centering
    \includegraphics[width=1\textwidth]{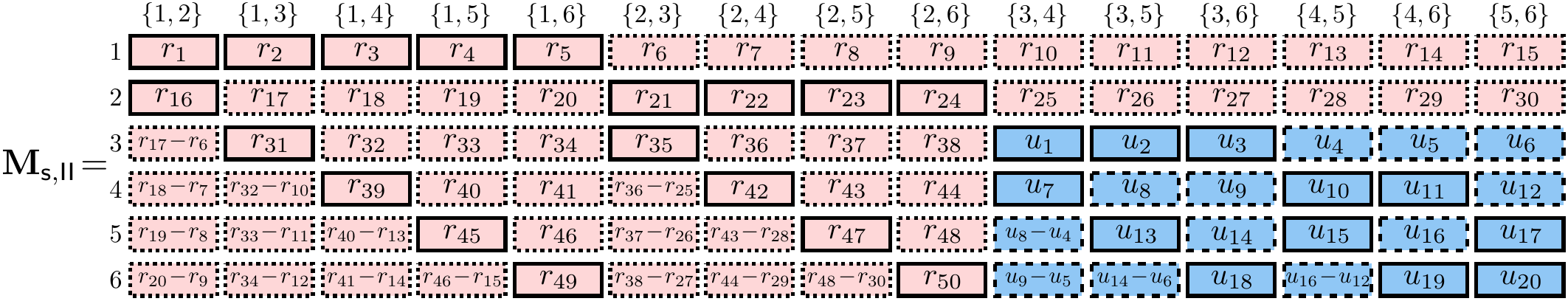}
    \caption{
    Message matrix $\MsII$ (with row and column labels) for a Type-II secure determinant code with parameters $(n,k,d,\ell) = (n,6,6,2)$, and mode $m = 2$. 
    Random keys are denoted by $r$ (in pink boxes), while secure symbols are denoted by $u$ (in blue boxes). 
    According to block matrix decomposition of $\MsII$ in~\eqref{eq:D-decompos}, parity symbols in pink dotted boxes are functions of random keys, while parity symbols in blue dashed boxes are functions of secure symbols.
    All parity symbols are placed such that the parity equations in~\eqref{eq:parity} are satisfied.
    }
    \label{fig:ex_secure_II}
    \vspace{5pt}
    \hrule
\end{figure*}

\subsection{Proof of Type-II Security Constraint of Theorem \ref{thm:achv_II}}
\label{sec:achvProof_II}

The proposed code construction is a secure version of the determinant code that is secure against Type-II eavesdroppers. Similar to the Type-I secure determinant code construction, it is evident that it maintains the \textit{Data Recovery} property due to~\cite[Proposition~1]{elyasi2016determinant}, as well as the \textit{Node Repair} property due to~\cite[Proposition~1]{elyasi2019determinant}. It remains to prove that the secure determinant code proposed in Section~\ref{sec:codeConstruct_II} satisfies the Type-II security constraint in~\eqref{req_mut_II}. To this end, we introduce three key lemmas essential for the proof of Type-II security property. 

\begin{lm}
    \label{lm:secProof_II_lm2}
    For every set of compromised nodes $\cL \subseteq [n]$ with $|\cL|\leq\ell$, the entropy of the eavesdropper's observation $\cEII(\cL)$ is upper bounded by the number of keys, i.e.,
    \begin{align*}
        H(\cEII(\cL))\leq |\key|, \qquad 
        \forall \cL \subseteq [n] \text{ with } |\cL|\leq\ell.
    \end{align*}
\end{lm}
The proof of Lemma~\ref{lm:secProof_II_lm2} is presented in Appendix~\ref{prf:lm:secProof_II_lm2}. 

\begin{lm}
    \label{lm:secProof_II_lm1}
    For the determinant code construction in Section~\ref{sec:codeConstruct_II}, for every subset of compromised nodes $\cL \subseteq [n]$ with $|\cL|=\ell$, the set of random keys can be fully recovered from the secure message $\cS$ and the eavesdropper's observation $\cEII(\cL)$, i.e.,
    \begin{align*}
        H(\key | \cEII(\cL),\cS)=0. 
    \end{align*}
\end{lm}
We refer to Appendix~\ref{prf:lm:secProof_II_lm1} for the proof of Lemma~\ref{lm:secProof_II_lm1}. 

Now, we are ready to prove that the proposed coded construction satisfies the Type-II security constraint in~\eqref{req_mut_II} as follows. First, note that if less than $\ell$ nodes are compromised, we can enhance the eavesdropper by providing her with the incoming data to $\ell-|\cL|$ nodes. Therefore, without loss of generality, we may assume  $|\cL|=\ell$. Thus, 
\begin{align}
    I(\cS;\cEII(\cL)) & = H(\cEII(\cL))-H(\cEII(\cL)|\cS) 
    \nonumber\\
    & \stackrel{\rm{(a)}}{\leq} |\key|-H(\cEII(\cL)|\cS) 
    \nonumber\\
    & \leq |\key| - H(\cEII(\cL)|\cS) + H(\cEII(\cL)|\cS,\key) 
    \nonumber\\
    & = |\key|-I(\cEII(\cL); \key |\cS) 
    \nonumber\\
    & = |\key| - H(\key|\cS) + H(\key|\cEII(\cL),\cS) 
    \nonumber\\
    & \stackrel{\rm{(b)}}{=} |\key|-H(\key|\cS) \nonumber\\
    & \stackrel{\rm{(c)}}{=} |\key|-|\key| = 0, \nonumber
\end{align}
where {\rm (a)} and~{\rm(b)} follow from Lemma~\ref{lm:secProof_II_lm2} and Lemma~\ref{lm:secProof_II_lm1}, respectively, and {\rm (c)} holds since the random keys are independent of the secure source symbols. This completes the proof of Theorem~\ref{thm:achv_II}. 
\hfill $\square$

\subsection{The Secrecy Capacity of Type-II Secure Determinant Codes
}
\label{sec:converseProof_II}
In this section, we present the proof of Property \ref{thm:converse_II} and provide a tight upper bound on the maximum file size to guarantee Type-II security for determinant codes. Consider an $(n,k=d,d,\ell)$ Type-II secure distributed storage system. Then, if $\cS$ is a file securely stored in the system, we have
\begin{align}
    & H\!\left(\cS\right)
    \nonumber\\
    &\stackrel{\rm{(a)}}{=} H\left(\cS\right) - I(\cS; \cEII(\cL))
    \nonumber\\
    & \!= 
    H\!\left(\cS\:|\:\cEII(\cL)\right)
    \nonumber\\
    & = 
    H\!\left(\cS \:|\: \{\R{i}{j}: j \in [\ell],\: i \in [n] \} \right)
    \nonumber\\
    & \leq
    H\!\left(\cS,\: \cN_{[d]}   \:|\: \{\R{i}{j}: j \in [\ell],\: i \in [n] \!\setminus\! j\}\right)
    \nonumber\\
    & = 
    H\left(\cN_{[d]} \:|\: \{\R{i}{j}: j \in [\ell],\: i \in [n] \}\right)
    \nonumber\\
    &\phantom{=}
    + H\left(\cS \:|\: \cN_{[d]},\:  \{\R{i}{j}: j \in [\ell],\: i \in [n] \}\right)
    \nonumber\\
    & \stackrel{\rm{(b)}}{=}
    H\!\left(\cN_{[\ell]}, \cN_{[\ell+1: \ell+m]}, \cN_{[\ell+m+1: d]}  | \{\R{i}{j}\!: j \!\in\! [\ell], i \!\in\! [n] \}\!\right)\!
    \nonumber\\
    & \stackrel{\rm{(c)}}{=}
    H\left( \cN_{[\ell+1: \ell+m]},\: \cN_{[\ell+m+1: d]}  \:|\: \{\R{i}{j}: j \in [\ell],\: i \in [n] \}\right)
    \nonumber\\
    & \stackrel{\rm{(d)}}{\leq} H\left( \cN_{[\ell+1: \ell+m]}, \cN_{[\ell+m+1: d]} , \cT \:|\: \{\R{i}{j}\!: j \!\in\! [\ell], i \!\in\! [n] \}\right)
    \nonumber\\
    &=  \!H\!\left( \cN_{[\ell+1: \ell+m]} \:| \{\R{i}{j}: j \in [\ell], i \in [n] \}\right)
    \nonumber\\
    &\phantom{=}
    +
    H\!\left(  \cT  \:| \cN_{[\ell+1: \ell+m]}, \{\R{i}{j}: j \in [\ell], i \in [n] \}\right) \nonumber\\
    &\phantom{=} + H\!\left(  \cN_{[\ell+m+1: d]}  \: |\:    \cN_{[\ell+1: \ell+m]}, \{\R{i}{j}\!: j \!\in\! [\ell], i \!\in\! [n] \} ,  \cT\right)\!,\!
    \nonumber\\
    \label{eq:lowerB_II}
\end{align}
where~{\rm (a)} follows from the secrecy constraint in~\eqref{req_mut_II}, the equality in~{\rm (b)} holds due to the data recovery property in \eqref{eq:recpvery} which implies that all the entries of the message matrix can be recovered from the content of any $d=k$ nodes, the equality in~{\rm (c)} follows form the node repair property in \eqref{eq:repair}, where the content of each node $j\in[\ell]$ can be retrieved from the repair data coming from all other nodes, and in~{\rm (d)} we introduce a tuple of random variables $\cT$ given by ${\cT \hspace{-2pt}\coloneqq\hspace{-1.5pt} \bigcup_{i=\ell+m+2}^{d+1} \bigcup_{j=\ell+m+1}^{i-1} \R{i}{j} \hspace{-1pt}=\hspace{-1.6pt} \bigcup_{j=\ell+m+1}^{d} \bigcup_{i=j+1}^{d+1} \R{i}{j}}$. 

Next, we bound each term in~\eqref{eq:lowerB_II}. For the first term, we can write
\begin{align}
    &H\left( \cN_{[\ell+1: \ell+m]} | \{\R{i}{j}: j \in [\ell], i \in [n] \}\right)
    \nonumber\\
    &\leq
    \sum_{u =\ell+1}^{\ell+m} H\left(\cN_{u} \:|\: \{\R{i}{j}: j \in [\ell],\: i \in [n]  \}\right)
    \nonumber\\
    &\leq
    \sum_{u =\ell+1}^{\ell+m}H\left(\cN_{u} \:|\: \{\R{u}{j}: j \in [\ell]\}\right)
    \nonumber
\end{align}
\begin{align}
    &= \sum_{u =\ell+1}^{\ell+m} \left[ H\left(\cN_{u}\right) + H(\{\R{u}{j}: j \in [\ell]\} \:|\: \cN_u) \right.
    \nonumber\\
    &\hspace{50pt}
    \left.
    - H\left( \{\R{u}{j}: j \in [\ell]\}\right)\right] \nonumber\\
    &\stackrel{\rm{(e)}}{=}
    \sum_{u =\ell+1}^{\ell+m} \lb \binom{d}{m} + 0 - \lp \binom{d}{m} - \binom{d-\ell}{m}\rp\rb 
    \nonumber\\
    &= \sum_{u =\ell+1}^{\ell+m} \binom{d-\ell}{m} = m \binom{d-\ell}{m},
    \label{eq:lowerB_II:T1}
\end{align}
where in~{\rm (e)} we used the fact that in a determinant code of mode $m$, the entropy of repair data sent from node $u$ to repair the nodes in $\cA$ satisfies \[H(\{ \R{u}{j}:j\in \cA\}) = \beta_{|\cA|}^{(m)} = \binom{d}{m} - \binom{d-|\cA|}{m},\] as proved in~\cite[Theorem~2]{elyasi2019determinant}. 

The second term in~\eqref{eq:lowerB_II} can be bounded as 
\begin{align}
    \label{eq:lowerB_II:T2}
    & H \left(  \cT  \:| \cN_{[\ell+1: \ell+m]}, \{\R{i}{j}: j \in [\ell], i \in [n] \}\right) 
    \nonumber\\
    &=
    H\Bigg(\bigcup_{u=\ell+m+2}^{d+1} \bigcup_{j=\ell+m+1}^{u-1}  \R{u}{j} \nonumber\\
    & \qquad\qquad\qquad\qquad\Bigg| \cN_{[\ell+1:\ell+m]}, \{\R{i}{j}: j \in [\ell], i \in [n]\}\Bigg)
    \nonumber\\
    &\leq
    \sum_{u=\ell+m+2}^{d+1} H\bigm(\{\R{u}{j} : j\in[\ell + m + 1:u - 1]\} \nonumber\\
    & \hspace{125pt} \bigm| \{\R{i}{j}: j \in [\ell], i \in [n] \}\bigm)
    \nonumber\\
    &\leq
    \!\!\!\!\sum_{u=\ell+m+2}^{d+1} \!\!\!\!H\!\left(\{\R{u}{j}\!: j\in[\ell\!+\!m\!+\!1:u\!-\!1]\} \big| \{\R{u}{j}\!:\! j \!\in\! [\ell]\}\right)
    \nonumber\\
    &=
    \sum_{u=\ell+m+2}^{d+1} \Big[ H\big(\{\R{u}{j}: j\in[\ell] \cup [\ell+m+1:u-1] \}\big)  
   \nonumber\\
       &\phantom{=======} 
    - H\big( \{\R{u}{j}: j\in[\ell]\}\big) \Big]
    \nonumber\\
    &\hspace{-3pt}\stackrel{\rm{(f)}}{=}
    \sum_{u=\ell+m+2}^{d+1} \!\lb\!
    \lp\!\!\binom{d}{m} \!-\! \binom{d\!+\!m+\!1\!-\!u}{m} \!\!\rp \!-\! \lp\!\! \binom{d}{m} \!-\! \binom{d\!-\!\ell}{m} \!\!\rp\!
    \rb\!\nonumber\\
    &= (d-\ell-m) \binom{d-\ell}{m} - \sum_{t=m}^{d-\ell-1} \binom{t}{m}
    \nonumber\\
    &= (m+1)\binom{d-\ell}{m+1} - \binom{d-\ell}{m+1} = m\binom{d-\ell}{m+1}, 
\end{align} 
where in {\rm (f)} we used the fact that \[H(\{ \R{u}{j}:j\in \cA\}) =  \binom{d}{m} - \binom{d-|\cA|}{m},\] which is proved in~\cite[Theorem~2]{elyasi2019determinant}. In order to bound the third term in~\eqref{eq:lowerB_II} we can write
\begin{align}
    &H\left(\cN_{[\ell+m+1:d]} \right. \left. \:|\: \{\R{i}{j}: j \in [\ell],\: i \in [n] \},  \cN_{[\ell+1:\ell+m]}, \cT \right)
    \nonumber\\
    &\stackrel{\rm{(g)}}{\leq}
    H\left(\cN_{[\ell+m+1:d]} \:|\: \cN_{[\ell]},  \cN_{[\ell+1:\ell+m]},\cT \right)
    \nonumber\\
    &=
    \sum_{u=\ell+m+1}^{d} 
    H\left(\cN_{u} \:|\:  \cN_{[\ell]}, \cN_{[\ell+1:\ell+m]},  \cN_{[\ell+m+1:u-1]}, \:\cT \right)
    \nonumber\\
    &\leq
    \!\!\!\!\!\sum_{u=\ell+m+1}^{d}
    \!\!\!\!\!\!\! H\left(\cN_{u} | \{\R{i}{u}\!: i\!\in\![u\!-\!1]\}, \{\R{i}{u}\!: i\!\in\! [u\!+\!1\!:\!d\!+\!1]\}  \right)
    \nonumber\\
    &\stackrel{\rm{(h)}}{=} 0,\label{eq:lowerB_II:T3}
\end{align} 
where~{\rm (g)} follows from the node repair property in \eqref{eq:repair}, which implies $\cN_{[\ell]}$ can be retrieved from the repair data coming to nodes in $[\ell]$, and similarly, we used the node repair property in {\rm (h)} to conclude that $\cN_u$ can be recovered from $d$ repair data coming from helper nodes in $[d+1]\setminus \{u\}$. Finally, plugging~\eqref{eq:lowerB_II:T1}--\eqref{eq:lowerB_II:T3} into~\eqref{eq:lowerB_II} we obtain
\begin{align*}
    H(\cS) \leq m\binom{d\!-\!\ell}{m} + m \binom{d\!-\!\ell}{m\!+\!1} = m\binom{d\!-\!\ell\!+\!1}{m\!+\!1} = \FsII^{(m)}.
\end{align*}
This completes the proof of Property~\ref{thm:converse_II}.
\hfill $\square$
    
\section{Conclusion}
\label{sec:conclusion}
In this paper, we develop information-theoretic secure determinant codes against Type-I and Type-II eavesdroppers. For system parameters $(n,k=d,d)$, we provide code constructions and characterize the achievable trade-offs for \mbox{Type-I} and Type-II secure determinant codes. We show that the proposed code constructions data recovery and node repair properties, along with the security constraints. Finally, we prove that the proposed construction is optimal, within  the  class of determinant codes. The general proof of optimality (without a constraint in the construction scheme) remains open for future works. Another related research problem is to develop secure codes for general $(n,k,d)$ parameters. We believe such a construction can be obtained using the non-secure cascade codes proposed in \cite{elyasi2020cascade}. However, the details of the construction and proof secrecy are not straightforward. Another interesting research direction is to prove whether the proposed secure determinant codes are optimal over all secure exact-repair regenerating DSS codes with parameters $(n,k,d)$.

\appendices

\section{Proof of Property~\ref{cor:num_pareto}}
\label{app:pareto}
In this section, we prove Property~\ref{cor:num_pareto}, in which the number of Pareto optimum points of the achievable region of \mbox{Type-II} secure determinant codes is characterized. Recall that Theorem~\ref{thm:achv_II} provides a set of $d$ achievable tuples \[\left\{\left(\alpha^{(m)}, \beta^{(m)}, \FsII^{(m)}, \right): m \in [d]\right\},\] or equivalently a set of $d$ achievable normalized pairs 
\[
\left\{\left(\bar{\alpha}{(m)}, \bar{\beta}^{(m)}\right) =  \left( \frac{\alpha^{(m)}}{ \FsII^{(m)}} ,  \frac{\beta^{(m)}}{ \FsII^{(m)}}  \right): m \in [d]\right\},
\] 
and any point in the convex hull of these pairs is achievable. However, not all of these $d$ points lie on the boundary of the achievable region, and some of them can be interior points of the region. Our goal is to characterize the exact number of corner points on the boundary of the achievable region. 

We call an achievable point of mode $m$ with parameters $\left(\bar{\alpha}^{(m)},\bar{\beta}^{(m)}\right)$ a \textit{Pareto point} if it is on the boundary of the achievable region, and call it an \textit{interior point} otherwise. In other words, an interior point is a pair $\left(\bar{\alpha}, \bar{\beta}\right)$ where each parameter is greater than or equal to an affine combination of the corresponding parameter of some Pareto points. In the single Pareto point, the only active corner point is the MBR point ($m=1$), which was shown in \cite{ye2017rate} and \cite{tandon2016toward}. Next, we examine the case of multiple Pareto points.

First consider $\bar{\beta}^{(m)}$ and $\bar{\beta}^{(m+1)}$. We have 
\begin{align*}
    \frac{\bar{\beta}^{(m)}}{\bar{\beta}^{(m+1)}} 
    &= \frac{{\beta^{(m)}} / {\FsII^{(m)}} }{ {\beta^{(m+1)}}/{\FsII^{(m+1)}} } 
    \nonumber\\
    &= \frac{\binom{d-1}{m-1}/m\binom{d-\ell+1}{m+1} }{ \binom{d-1}{m}/(m+1)\binom{d-\ell+1}{m+2}} \\
    &= \frac{m+1}{m+2} \frac{d-m-\ell}{d-m} <1.
\end{align*}
This implies that $\bar{\beta}^{(m)} < \bar{\beta}^{(m+1)}$, i.e., $\bar{\beta}^{(m)}$ is increasing with respect to $m$, and hence $m=1$ provides the lowest value of $\bar{\beta}$. Consequently, $\bar{\beta}^{(1)}$ is always a Pareto point. 

Next, assume both $\left(\bar{\alpha}^{(m)},\bar{\beta}^{(m)}\right)$ and $\left(\bar{\alpha}^{(m+1)},\bar{\beta}^{(m+1)}\right)$ are Pareto points. Then, they should satisfy  $\bar{\alpha}^{(m)} > \bar{\alpha}^{(m+1)}$, otherwise $\left(\bar{\alpha}^{(m+1)},\bar{\beta}^{(m+1)}\right) \geq \left(\bar{\alpha}^{(m)},\bar{\beta}^{(m)}\right)$, which is in contradiction with $\left(\bar{\alpha}^{(m+1)},\bar{\beta}^{(m+1)}\right)$  being a Pareto point. This implies that
\begin{align}
    1 < \frac{\bar{\alpha}^{(m)} }{\bar{\alpha}^{(m+1)}} 
    & = \frac{{\alpha^{(m)}} / {\FsII^{(m)}} }{ {\alpha^{(m+1)}}/{\FsII^{(m+1)}} } 
    \nonumber\\
    & = \frac{\binom{d}{m}/m\binom{d-\ell+1}{m+1} }{ \binom{d}{m+1}/(m+1)\binom{d-\ell+1}{m+2}} 
    \nonumber\\
    & = \frac{(m+1)^2 (d-m-\ell)}{m(m+2)(d-m)}.
    \label{eq:pareto_inequality}
\end{align}
Then,~\eqref{eq:pareto_inequality} holds if and only if 
\begin{align}
    \ell &< \frac{d-m}{(m+1)^2}\left[(m+1)^2 -m(m+1) \right]\nonumber\\ 
    &= \frac{d+1-(m+1)}{(m+1)^2}. 
    \label{eq:pareto_bound_ell}
\end{align}
Note that the RHS of~\eqref{eq:pareto_bound_ell} is a decreasing function of $m$, and hence, if it is not satisfied for $m$, then it will not hold for $m+1$. In other words, the set of Pareto points are those corresponding to $\{1,2,\dots, t\}$, where $t$ is the largest integer satisfying $\ell < \frac{d+1-t}{t^2}$. Solving the quadratic equation $t$, we can conclude that $t$ is the largest integer that satisfies
\begin{align}
    t < \frac{-1 + \sqrt{1+4\ell (d+1)}}{2\ell}, 
    \label{eq:pareto_condition}
\end{align}
which is the claim of the property. 

Note that~\eqref{eq:pareto_condition} reduces to the result in \cite{ye2019secure} for a single Pareto point. In order to have \emph{only} one Pareto point, we need that $t=2$ violates the condition in~\eqref{eq:pareto_condition}. This implies ${4\ell+1 \geq \sqrt{1+4\ell(d+1)}}$, or equivalently, 
\begin{align*}
    \ell \geq \frac{d-1}{4},
\end{align*}
which subsumes the result of \cite[Theorem 1]{ye2019secure}.  \hfill $\square$

\section{Proof of Lemma~\ref{lm:secProof_I_lm2}} 
\label{prf:lm:secProof_I_lm2}
    Let $\cL=\{e_1,e_2,\dots, e_{|\cL|}\}$ be the set of nodes accessed by the eavesdropper. Then, using the chain rule we can write
    \begin{align}\label{eq:prf:lm1:1}
        H&(\cEI(\cL)) = H(\cN_{e_1}, \cN_{e_2}, \dots, \cN_{e_|\cL|}) \nonumber\\
        &= H(\cN_{e_1},   \dots, \cN_{e_m})  +\!\! \sum_{j=m+1}^{|\cL|}H(\cN_{e_j}| \cN_{e_1},   \dots, \cN_{e_{j-1}}) \nonumber\\
        &\leq \sum_{j=1}^m H(\cN_{e_j}) +\!\! \sum_{j=m+1}^{|\cL|} H(\cN_{e_j}|  \cN_{e_1},   \dots, \cN_{e_{j-1}}).
    \end{align}
Let us focus on each term in the summation in~\eqref{eq:prf:lm1:1}. Fix some $j\in[m+1:|\cL|]$, and consider a set of nodes ${\cP\subseteq [n]\setminus \{e_1,e_2,\dots, e_j\}}$ with $|\cP|=d-(j-1)$, and let $\cH=\cP \cup \{e_1,e_2,\dots, e_{j-1}\}$. Since $|\cH|=d$ and $e_j \notin \cH$, the content of node $e_j$ can be exactly repaired by the repair data sent from nodes in $\cH$. Using~\cite[Proposition~2]{elyasi2016determinant}, we have
\begin{align}\label{eq:prf:lm1:2}
    \cN_{e_j}  &= \X{\cH}{e_j} \mathbf{\Theta}_{e_j}^{\cH},
\end{align}
where $\mathbf{\Theta}_{e_j}^{\cH}$ is a $\binom{d}{m}\times\binom{d}{m}$ matrix,  whose entries only depend on the encoder matrix $\mathbf{\Psi}$, and $\X{\cH}{e_j}$ is a row vector of length~$\binom{d}{m}$, where its entries are labeled by subsets of $\cH$ of size $m$. In particular, the entry at position $\cI$ (with $\cI\subseteq \cH$ and $|\cI|=m$) of $\X{\cH}{e_j}$ is given by~\cite[Proposition~2]{elyasi2016determinant}
\begin{align}\label{eq:prf:lm1:3}
    \X{\cH}{e_j}(\cI) \hspace{-1pt}=\hspace{-2pt} \sum_{i\in \cI} \sum_{\substack{\cJ\subseteq [d]\\ |\cJ|=m}} \hspace{-3pt}\cN_{i}(\cJ) \hspace{-1pt}\cdot\hspace{-1pt} \det{\mathbf{\Psi} (\cI \hspace{-1pt}\cup\hspace{-1pt}\{e_j\}\hspace{-1pt}\setminus\hspace{-1pt}\{i\} , \cJ)}.
\end{align}
Here, $\cN_{i}(\cJ)$ is the $\cJ$th coded symbol stored in node $i$, and $\mathbf{\Psi} (\cI \cup\{e_j\}\setminus\{i\} , \cJ)$ is an $m\times m$ submatrix of $\mathbf{\Psi}$ obtained by the set of rows in~$\cI \cup\{e_j\}\setminus\{i\}$ and the set of columns in~$\cJ$. Therefore, we can write

\begin{align}\label{eq:prf:lm1:4}
     H(\cN_{e_j} |  &\cN_{e_1}, \cN_{e_2}, \dots, \cN_{e_{j-1}}) 
     \nonumber\\
     &\stackrel{\rm{(a)}}{\leq}   H(\X{e_j}{H}|  \cN_{e_1}, \cN_{e_2}, \dots, \cN_{e_{j-1}}) 
     \nonumber\\
     &\leq \sum_{\substack{\cI\subseteq \cH\\ |\cI|=m}} H(\X{\cH}{e_j}(\cI)|  \cN_{e_1}, \cN_{e_1}, \dots, \cN_{e_{j-1}})
     \nonumber\\
     &= \sum_{\substack{\cI\subseteq \cH\\ |\cI|=m \\ \cI\subseteq \{e_1,\dots, e_{j-1}\}}} \hspace{-5mm}H(\X{\cH}{e_j}(\cI)|  \cN_{e_1}, \cN_{e_1}, \dots, \cN_{e_{j-1}})
     \nonumber\\
     & \qquad + \hspace{-19pt}
     \sum_{\substack{\cI\subseteq \cH\\ |\cI|=m\\ \cI \nsubseteq  \{e_1,\dots, e_{j-1}\} }} \hspace{-5mm}H(\X{\cH}{e_j}(\cI)|  \cN_{e_1}, \cN_{e_1}, \dots, \cN_{e_{j-1}})\nonumber\\
     &\stackrel{\rm{(b)}}{\leq}  0 + \binom{d}{m} - \binom{j-1}{m} = \alpha - \binom{j-1}{m},
\end{align}
where {\rm (a)} follows from~\eqref{eq:prf:lm1:2} and the fact that the encoder matrix $\mathbf{\Psi}$ is a public information. We have to consider two cases for {\rm (b)}: If $\cI\hspace{-1pt}\subseteq\hspace{-1pt}\{e_1,\dots, e_{j-1}\}$, then~\eqref{eq:prf:lm1:3} implies that $\X{\cH}{e_j}(\cI)$ is a deterministic function of $(\cN_{e_1},\dots, \cN_{e_{j-1}})$, and hence,  $H(\X{\cH}{e_j}(\cI)|  \cN_{e_1},  \dots, \cN_{e_{j-1}})=0$. Moreover, when ${\cI\nsubseteq\{e_1,\dots, e_{j-1}\}}$, then the conditional entropy of $\X{\cH}{e_j}(\cI)$ is at most~$1$. These lead to the inequality in~{\rm (b)}. Plugging~\eqref{eq:prf:lm1:4} into~\eqref{eq:prf:lm1:1}, we get

\begin{align}
        H(\cEI(\cL))  
        &\leq \sum_{j=1}^m H(\cN_{e_j}) 
        + \sum_{j=m+1}^{|\cL|} H(\cN_{e_j}|  \cN_{e_1},  \dots, \cN_{e_{j-1}})\nonumber\\
        & \leq \sum_{j=1}^m \alpha + \sum_{j=m+1}^{|\cL|} \left[ \alpha -\binom{d}{j-1}\right]\nonumber\\
        &= |\cL| \alpha - \sum_{j=m+1}^{|\cL|} \binom{j-1}{m}\nonumber\\
        & =  |\cL| \alpha - \binom{|\cL|}{m+1}\nonumber\\
        &\leq \ell \alpha - \binom{\ell}{m+1} = |\cQ|.\nonumber
\end{align}
This completes the proof of Lemma~\ref{lm:secProof_I_lm2}. \hfill $\square$

\section{Proof of Lemma~\ref{lm:secProof_I_lm1}}
\label{prf:lm:secProof_I_lm1}
Recall that ${\cN_i=\mat{\mathbf{\Psi}}{i,:} \MsI}$ denotes the contents of node~$i$ for $i \in [n]$. Moreover, for every $\cL \subseteq [n]$ and $|\cL|=\ell$, let $\cEI(\cL)$ be the data observed by the eavesdropper, that is, the content of all nodes $i\in \cL$. We can stack all such rows in a matrix, to  construct $\bEI = \mat{\bC_\mathsf{s,I}}{\cL,:} =\mat{\mathbf{\Psi}}{\cL,:}\cdot\MsI$, where $\mat{\bC_\mathsf{s,I}}{\cL,:}$ and $\mat{\mathbf{\Psi}}{\cL,:}$ are, respectively, submatrices of $\bC_\mathsf{s,I}$ and~$\mathbf{\Psi}$ generated by all columns and only rows with indices belong to $\cL$. Recall that each column of $\MsI$ and $\bEI$ is indexed by a subset $\cX$ where $\cX \subseteq [d]$ and $|\cX|=m$. Let ${\mat{\MsI}{:,\cX}}$ and $\bEI(:,\cX)$ denote the $\cX$th column of $\MsI$ and $\bEI$, respectively. For fixed parameters $(d,\ell)$, let $\mat{\MsIOv}{:,\cX}$ be an $\ell \times 1$ column vector that includes the top $\ell$ entries of ${\mat{\MsI}{:,\cX}}$, i.e., ${\mat{\MsIOv}{:,\cX}=\mat{\MsI}{[1:\ell],\cX}}$, and $\mat{\MsIUn}{:,\cX}$ be a $(d-\ell)\times 1$ column vector that consists the bottom $(d-\ell)$ entries of $\mat{\MsI}{:,\cX}$, that is, $\mat{\MsIUn}{:,\cX}= \mat{\MsI}{[\ell+1:d],\cX}$. Therefore, the $\cX$th column of $\bEI$ can be written as 
\begin{align}
    \bEI(:,\cX) 
    &= 
    \mat{\mathbf{\Psi}}{\cL,:} \: \mat{\MsI}{:,\cX}
    \nonumber\\
    &= \mat{\mathbf{\Psi}}{\cL,[1:\ell]} \: \mat{\MsIOv}{:,\cX} 
    \nonumber\\
    &\phantom{=} +
    \mat{\mathbf{\Psi}}{\cL,[\ell+1:d]} \: \mat{\MsIUn}{:,\cX}. 
    \label{eq:eavs-col}
\end{align} 
In order to show~\eqref{eq:I:Q-from-E&S}, we decode the random keys from the secure message and the eavesdropper's observation by reconstructing the message matrix $\MsI$, from which all key symbols can be retrieved. The reconstruction of $\MsI$ is performed column-by-column, in a recursive manner, in \emph{reverse lexicographical} order of the column labels (i.e., from right to left). More precisely, we start with the last column with index ${[d-m+1:d]=\{d-m+1, d-m+2, \dots, d\}}$ and decode its entries. Due to the order of the reconstruction, by the time we start decoding column $\cX$,  all columns $\MsI(:,\cY)$ with $\cY \succ \cX$ are already decoded. 

Now, we can expand $H(\key|\cEI(\cL),\cS)$ as  
\begin{align}
    &H(\key \:|\: \cEI(\cL),\cS)
    \nonumber\\
    & = 
    H(\key \:|\: \bEI,\cS)
    \nonumber\\
    & = 
    H\lp\MsI \:\Big|\: \bEI,\cS\rp
    \nonumber\\
    & =
    H\left(\{\mat{\MsI}{:,\cX} : \cX\subseteq [d], |\cX|=m\} \:\Big|\: \bEI, \cS\right)
    \nonumber\\
    & \stackrel{\rm{(a)}}{=}
    \sum_{\substack{\cX\subseteq [d],\\ |\cX|=m } } 
    H\lp\mat{\MsI}{:,\cX} \:|\: \bEI, \cS, \{\mat{\MsI}{:,\cY} : \cY \succ \cX\} \rp
    \nonumber\\
    & \stackrel{\rm{(b)}}{=}
    \sum_{\substack{\cX\subseteq [d],\\ |\cX|=m } }
    \Bigg[ H(\mat{\MsIUn}{:,\cX} \:|\: \bEI, \cS, \{\mat{\MsI}{:,\cY}: \cY \succ \cX\} )  
    \nonumber\\
    & \phantom{=}
    +\!  H\!\!\left( \MsIOv(:,\cX) \Big| \bEI, S, \{\MsI(:,\cY): \cY \succ \cX\}, \MsIUn(:,\cX) \right) \!\!\Bigg]\!,\!\!
    \label{eq:secProof_decode2}
\end{align}
where~{\rm (a)} and~{\rm (b)} follow from the chain rule, and the fact that ${\MsI(:,\cX) = \left\{\MsIOv(:,\cX), \MsIUn(:,\cX)\right\}}$. 
Next, we show that each term in the summation in \eqref{eq:secProof_decode2} is equal to zero.

Note that $\MsIUn(:,\cX)=\lc\MsI(x,\cX) : x \in [\ell+1:d]\rc$, and recall from Section~\ref{sec:codeConstruct_I} that the entries in the bottom~${(d-\ell)}$ rows of $\MsI$ are either secure source symbols or parity symbols. Hence, each entry of $\MsIUn(x,\cX)$ can be categorized into three groups as follows:
\begin{itemize}
    \item If $x\in \cX$, then $(x,\cX)\in \cV(\MsI)$ and thus, $\mat{\MsI}{x,\cX}$ is a secure symbol. This implies   $H(\MsI(x,\cX) | \cS)=0$. 
    \item If $x\notin \cX$ and $x<\max \cX$, then ${(x,\cX)\in\cW(\MsI)}$. This implies tha $\MsI(x,\cX)$ is a secure symbol, and hence,  ${H(\MsI(x,\cX)|\cS)=0}$.
    \item Finally, when $x\notin \cX$ and $x>\max \cX$ we have  ${(x,\cX)\in\cP(\MsI)}$, and thus, $\MsI(x,\cX)$ is a parity symbol.  The parity equation~\eqref{eq:parity} for the parity group $ \cX\cup \{x\}$ implies that 
    \begin{align}
        \MsI&(x,\cX)\nonumber\\
        &=  (-1)^{m}\hspace{-1pt} \sum_{y\in \cX } 
        (-1)^{\ind{\cX}{y}} \MsI(y,(\cX \cup \{x\})\!\setminus\! \{y\}).
        \nonumber
    \end{align}
 Note that for every ${y\in \cX}$ we have ${y\leq \max \cX <x}$, which implies ${\cY=(\cX\cup\{x\})\setminus \{y\} \succ \cX}$. Therefore, all symbols $\MsI(y,(\cX \cup \{x\})\!\setminus\! \{y\})$ with ${y\in \cX}$ appear in $\{\mat{\MsI}{:,\cY}: \cY \succ \cX\}$, and thus $\MsI(x,\cX)$ can be evaluated from the variables in the condition of the entropy expression. That is, 
 \[{H(\MsI(x,\cX) | \{\mat{\MsI}{:,\cY}: \cY \succ \cX\})=0}.\] 
\end{itemize}
This can be formalized as
\begin{align}
    &H\lp\mat{\MsIUn}{:,\cX} \:\cnd\: \bEI, \cS, \{\mat{\MsI}{:,\cY}: \cY \succ \cX\} \rp \nonumber\\
    &= \sum_{x\in[\ell+1:d]} H\lp\mat{\MsIUn}{x,\cX} \:\cnd\: \bEI, \cS, \{\mat{\MsI}{:,\cY}: \cY \succ \cX\} \rp \nonumber
\end{align}
\begin{align}
    &\leq \!\!\!\sum_{\substack{x\in[\ell+1:d] \\ x\leq \max \cX}} \!\! H\lp\mat{\MsIUn}{x,\cX} \:\cnd\: \bEI, \cS, \{\mat{\MsI}{:,\cY}: \cY\! \succ\! \cX\} \rp  
    \nonumber\\
    &\phantom{\leq}  + \!\!\!\sum_{\substack{x\in[\ell+1:d]\\ x> \max \cX}} \!\!H\lp\mat{\MsIUn}{x,\cX} \:\cnd\: \bEI, \cS, \{\mat{\MsI}{:,\cY}: \cY \!\succ\! \cX\} \rp 
    \nonumber\\
    &\leq \!\!\!\!\sum_{\substack{x\in[\ell+1:d]\\ x\leq \max \cX}} \!\!\!H\hspace{-2pt}\lp\mat{\MsIUn}{x,\cX} \!\:\cnd\:\! \cS \rp  \nonumber\\
    & \phantom{\leq} +\!\!\!\!\!\sum_{\substack{x\in[\ell+1:d]\\ x> \max \cX}} \!\!\!\!\! H\hspace{-2pt}\lp\mat{\MsIUn}{x,\cX} \!\:\cnd\:\!
    \{\mat{\MsIUn}{y,(\cX\hspace{-1pt}\cup\hspace{-1pt}\{x\})\hspace{-2pt}\setminus\hspace{-2pt}\{y\}}\hspace{-1pt}:\hspace{-1pt} y\in \cX\} \hspace{-1pt} \rp \nonumber\\
   & =0, 
   \label{eq:decode:lower}
\end{align}
which implies that the first term in the summation in \eqref{eq:secProof_decode2} is zero.

Next, recall from Condition~\ref{cond:psi:2} that $\mathbf{\Psi}(\cL,[1:\ell])$ is full-rank, and thus, invertible. This, together with~\eqref{eq:eavs-col}, implies that 
\begin{align*}
    \MsIOv(:,\cX) 
    &\hspace{-2pt}=\hspace{-2pt} \mathbf{\Psi}^{-1}({\cL,[1\hspace{-2pt}:\hspace{-2pt}\ell]}) \left(\bEI(:,\cX) \hspace{-1pt}-\hspace{-1pt} \mathbf{\Psi}_{\cL,[\ell+1:d]}  \MsIUn(:,\cX) \right)\hspace{-1pt},
\end{align*}
i.e., one can decode $\MsIOv(:,\cX)$ from $\MsIUn(:,\cX)$ and the eavesdropper observation. Therefore, 
\begin{align}
    & H\Big(\MsIOv(:,\cX) \:\cnd\: \bEI, \cS,\{\MsI(:,\cY): \cY \succ \cX\}, \MsIUn(:,\cX)\Big)
    \nonumber\\
    & \leq H\lp\MsIOv(:,\cX),  \:\cnd\: \bEI, \MsIUn(:,\cX)\rp
    = 0,
    \label{eq:decode:upper}
\end{align}
for every $\cX\subseteq [d]$ with $|\cX|=m$. Plugging \eqref{eq:decode:lower} and \eqref{eq:decode:upper} into \eqref{eq:secProof_decode2}, we conclude that $H(\key \:|\: \cEI(\cL),\cS) = 0$. This completes the proof of Lemma~\ref{lm:secProof_I_lm1}. 
\hfill $\square$

\section{Proof of Lemma~\ref{lm:secProof_II_lm2}}
\label{prf:lm:secProof_II_lm2}
Recall from~\eqref{eq:repair-data} that the repair data sent from  a helper node $h$ to a compromised node $f\in \cL$ can be determined by ${\R{h}{f}= \mat{\mathbf{\Psi}}{h,:}\cdot\MsII\cdot\bxi^{f}}$. Concatenating all such vector for all possible helpers and every $f\in \cL$, we arrive at $\mathbf{\Psi}\cdot\MsII\cdot\bxi^{\cL}$, where $\bxi^\cL$  is a matrix of size $\binom{d}{m} \times \ell \binom{d}{m-1}$ obtained by concatenating matrices $\bXi^f$ for all $f\in \cL$. Note that $\mathbf{\Psi}$ is a tall matrix, and all its $d\times d$ submatrices are full-rank (by Condition~\ref{cond:psi:1}). Hence, there is a one-to-one mapping between $\cEII(\cL)$ and $\MsII\cdot\bxi^{\cL}$. Next, note that even the matrix product $\MsII \bxi^\cL$ has $d$ rows and $\ell \binom{d}{m-1}$ columns, there are some linear dependencies among its entry. It is shown in~\cite[Theorem 3]{elyasi2019determinant}  that among all the repair data incoming to a set of $|\cL|$ failed nodes, only $m\binom{d+1}{m+1} - m\binom{d-|\cL|+1}{m+1}$ symbols are informative linearly independent. This immediately implies that 
\begin{align}
    H(\cEII(\cL)) &= H(\MsII\cdot\bxi^{\cL}) \nonumber\\
    &= m\binom{d+1}{m+1} - m\binom{d-|\cL|+1}{m+1} \nonumber\\
    &\leq 
    m\binom{d+1}{m+1} - m\binom{d-\ell+1}{m+1} =
    |\key|.
\end{align}
This completes the proof of Lemma~\ref{lm:secProof_II_lm2}.\hfill $\square$

\section{Proof of Lemma~\ref{lm:secProof_II_lm1}} 
\label{prf:lm:secProof_II_lm1}
Before we prove the desired claim, we present the following lemma, which plays an important role in characterizing the amount of data observed by the eavesdroppers. We present the proof of this lemma in Appendix~\ref{app:fullRank}.
\begin{lm}
    \label{prop:fullrank_uxi}
        Let $\bxi^\cL$ be the concatenation of all matrices ${\{\bxi^f: f\in \cL\}}$, where $\bxi^f$ is the repair encoder matrix introduced in~\eqref{eq:xi-def}, $\cL$ is an arbitrary set of $\ell=|\cL|$  distinct nodes, and $\ubxi^\cL$ be the submatrix of $\bxi^\cL$ obtained from its top $\binom{d}{m}-\binom{d-\ell}{m}$ rows. Then $\ubxi^\cL$  is full-rank, i.e.,  ${\rk{\ubxi^{\cL}}= \binom{d}{m}-\binom{d-\ell}{m}}$.
\end{lm}
 
Now, we are ready to prove Lemma~\ref{lm:secProof_II_lm1}. 
For every $\cL \subseteq [n]$ with  $|\cL|=\ell$, we can upper bound $H(\key|\cEII(\cL),\cS)$ as 
\begin{align}
     H(\key|&\cEII(\cL),\cS) 
    \nonumber\\
    &\leq H(\key, \MsII |\cEII(\cL),\cS)
    \nonumber\\
    &= H(\MsII|\cEII(\cL), \cS) + H(\key|\MsII, \cEII(\cL), \cS)
    \nonumber\\
    &  \stackrel{\rm{(a)}}{=}  H(\MsII|\cEII(\cL),\cS)
    \nonumber\\
    &  \stackrel{\rm{(b)}}{=}  H(\bA,\bB,\bC,\bD|\cEII(\cL),S) 
    \nonumber\\
    &  \stackrel{\rm{(c)}}{=}  H(\bD|\cEII(\cL),\cS) + H(\bC|\bD, \cEII(\cL),\cS) 
    \nonumber\\
    &\phantom{=} + H(\bA,\bB|\bC, \bD, \cEII(\cL),\cS),
    \label{eq:H(J|E,S)}
\end{align}
where~{\rm (a)} follows from the fact that all secure symbols and random keys can be retrieved from the message matrix $\MsII$; in~{\rm (b)} we replaced $\MsII$ by its  block decomposition given in~\eqref{eq:D-decompos}; and \eqref{eq:H(J|E,S)} follows from the chain rule. Next, we show that each term in \eqref{eq:H(J|E,S)} equals to zero.

For the first term in \eqref{eq:H(J|E,S)}, recall that all entries of $\bD$ in either $\cV$-type or $\cW$-type position is a secure information symbol, which is known given $\cS$. Also, Remark~\ref{rmk:bD} and~\eqref{eq:parity-Ds} imply that the $\cP$-type symbols with matrix  $\bD$ only depend on the symbols in $\cS$.  Therefore, we have 
\begin{align}
    H(\bD|\cEII(\cL),\cS) \leq H(\bD|\cS) =0. 
    \label{eq:1stTerm_secProof_II}
\end{align}

For the second term in \eqref{eq:H(J|E,S)}, recall that any set of compromised nodes  $\cL\subseteq [n]$ accessed by the eavesdropper, with $|\cL|=\ell$, the observation of the eavesdropper is of the form $\cEII(\cL) = \{\R{h}{f}: f\in \cL, h\in [n] \}$, where ${\R{h}{f}= \mat{\mathbf{\Psi}}{h,:}\cdot\MsII\cdot\bxi^{f}}$ is a vector of length $\beta=\binom{d}{m-1}$. For each helper node $h\in \cH$, concatenating all such vector for all $f\in \cL$, we arrive at $\mat{\mathbf{\Psi}}{h,:}\cdot\MsII\cdot\bxi^{\cL}$, where $\bxi^\cL$ is a matrix of size $\binom{d}{m} \times \ell \binom{d}{m-1}$ obtained by concatenating matrices $\bXi^f$ for all $f\in \cL$. Then, for a set of helper nodes $\cH\subseteq [n]$ with $|\cH|=d$, we can stack the repair data going from $h$ to all the nodes in $\cL$, and obtain a $d \times \ell\binom{d}{m-1}$ block matrix $\bX$,  given by
\begin{align}\label{eq:X-from-EII}
    \bX
    = \mathbf{\Psi}(\cH,:)  \cdot  \MsII \cdot \bxi^\cL, 
\end{align}
where $\mathbf{\Psi}(\cH,:)$ is a submatrix of $\mathbf{\Psi}$ of size  $d\times d$ obtained by rows $h\in \cH$. Since $\mathbf{\Psi}(\cH,:)$ is a Vandermonde matrix, it is full-rank, and hence invertible.  
Therefore, we have 
\begin{align}
    \mathbf{\Psi}^{-1}(\cH,:) \cdot \bX = \MsII \cdot \bxi^\cL.
\end{align}
Then, we decompose matrix $\bxi^\cL$  into two submatrices. We denote the submatrix of $\bxi^\cL$ consisting of the top $\binom{d}{m} - \binom{d-\ell}{m}$ rows of $\bxi^\cL$ by $\ubxi^\cL$, and the submatrix of $\bxi^\cL$ consisting of the bottom~$\binom{d-\ell}{m}$ rows of $\bxi^\cL$ by $\dbxi^\cL$. This allows us to write 
\begin{align} 
    \mathbf{\Psi}^{-1}(\cH,:)\cdot \bX 
     = \MsII \cdot \bxi^\cL 
    &=  \left[
    \begin{array}{c|c}
        \bA & \bB \\
        \hline
        \bC & \bD 
    \end{array}\right] \cdot \left[
    \begin{array}{c}
        \ubxi^{\cL} \\
        \hline
        \dbxi^{\cL}
    \end{array} \right]\nonumber\\
    &= \left[
    \begin{array}{c}
        \bA\ubxi^{\cL}+\bB \dbxi^{\cL}   \\
        \hline
        \bC\ubxi^{\cL}+\bD \dbxi^{\cL} 
    \end{array}
    \right].
  \label{eq:expand}
\end{align}
Therefore, we can write 
\begin{align}\label{eq:2ndTerm_secProof_II}
    &H(\bC|\bD, \cEII(\cL),\cS) 
    \nonumber\\
    &\leq H(\bC, \bX|\bD, \cEII(\cL),\cS) \nonumber\\
    &= H(\bX|\bD, \cEII(\cL),\cS) + H(\bC|\bX, \bD, \cEII(\cL),\cS)\nonumber\\
    &\leq H(\bX| \cEII(\cL)) + H(\bC, \bX|\bD, \cEII(\cL),\cS) \nonumber\\
    &\stackrel{\rm{(a)}}{\leq}  0 + H(\bC| \bA\ubxi^{\cL}+\bB \dbxi^{\cL}, \bC\ubxi^{\cL}+\bD \dbxi^{\cL}, \bD, \cEII(\cL),\cS)\nonumber\\
    &\stackrel{\rm{(b)}}{=} 0,
\end{align}
where~{\rm (a)} follows from~\eqref{eq:X-from-EII}  and~\eqref{eq:expand}, and~{\rm (b)} is due to the facts that $\bC = \left( (\bC\ubxi^{\cL}+\bD\dbxi^{\cL}) - \bD\dbxi^{\cL} \right)  \left(\ubxi^{\cL}\right)^{-1}$, and $\ubxi^{\cL}$ is full-rank, as shown in Lemma~\ref{prop:fullrank_uxi}.

Finally, we bound the third term in \eqref{eq:H(J|E,S)}. Let $\bEII(\cL)$ be an $\ell\times\alpha$ matrix obtained by stacking the content of the $\ell$ compromised nodes. We have   $\bEII(\cL) =\mathbf{\Psi}(\cL,:) \MsII$. Recall that the eavesdropper's observation is characterized by $\cEII(\cL)$, which is the incoming repair data to all nodes in $\cL$, and hence the eavesdropper can recover the content of the nodes in $\cL$, i.e., we have 
\begin{align}\label{eq:bEII-cEII}
    H(\bEII(\cL)|\cEII(\cL))=0.    
\end{align}

Next, using the decomposition of $\MsII$ in~\eqref{eq:D-decompos} we have
\begin{align}\label{eq:bEII-decompose}
    \bEII(\cL) &= 
    \mathbf{\Psi}(\cL,:) \cdot \MsII
    \nonumber\\
    &= \left[
    \begin{array}{c|c}
        \mathbf{\Psi}(\cL,[1:\ell]) &
        \mathbf{\Psi}(\cL,[\ell+1:d])
    \end{array} \right]
    \cdot \left[
    \begin{array}{c|c}
        \bA & \bB \\
        \hline
        \bC & \bD 
    \end{array}\right] 
    \nonumber\\ 
    &=
    \Big[
    \mathbf{\Psi}(\cL,[1:\ell]) \bA + \mathbf{\Psi}(\cL,[\ell+1:d]) \bC  
    \Big| \nonumber\\
    &\hspace{40pt} \mathbf{\Psi}(\cL,[1:\ell]) \bB + \mathbf{\Psi}(\cL,[\ell+1:d]) \bD 
    \Big],
\end{align}
where $\mathbf{\Psi}(\cL:[1:\ell])$ is a submatrix of $\mathbf{\Psi}$ obtained from the intersection of the rows with index in $\cL$ and the first $\ell$ columns, and $\mathbf{\Psi}(\cL:[\ell+1:d])$ is a submatrix of $\mathbf{\Psi}$ obtained from the intersection of rows with label in $\cL$ and the last~$(d-\ell)$ columns. On the other hand, we have     
\begin{align*}
    \bA = \mathbf{\Psi}^{-1}(\cL,[1:\ell]) \bigm[ &\lp \mathbf{\Psi}(\cL,[1:\ell]) \bA + \mathbf{\Psi}(\cL,[\ell+1:d]) \bC \rp \nonumber\\
    &- \mathbf{\Psi}(\cL,[\ell+1:d]) \bC \bigm],
\end{align*}
where $\mathbf{\Psi}(\cL,[1:\ell])$ is full-rank (due to Condition~\ref{cond:psi:2}), and $ \mathbf{\Psi}(\cL,[1:\ell]) \bA + \mathbf{\Psi}(\cL,[\ell+1:d]) \bC$ is given in the first $\binom{d}{m}-\binom{d-\ell}{m}$ columns of $\bEII(\cL)$ as shown in~\eqref{eq:bEII-decompose}.  This implies 
\begin{align}\label{eq:A-C-bEII}
    H(\bA | \bC, \bEII(\cL) )=0.
\end{align}
Similarly, since 
\begin{align*}
    \bB = \mathbf{\Psi}^{-1}(\cL,[1:\ell]) \bigm[ &\lp \mathbf{\Psi}(\cL,[1:\ell]) \bB + \mathbf{\Psi}(\cL,[\ell+1:d]) \bD \rp \nonumber\\
    &- \mathbf{\Psi}(\cL,[\ell+1:d]) \bD \bigm],
\end{align*}
and $\mathbf{\Psi}(\cL,[1:\ell]) \bB + \mathbf{\Psi}(\cL,[\ell+1:d]) \bD$ is given in the last $\binom{d-\ell}{m}$ columns of $\bEII(\cL)$, we get 
\begin{align}\label{eq:B-D-bEII}
    H(\bB | \bD, \bEII(\cL) )=0.
\end{align}

Therefore, from~\eqref{eq:bEII-cEII}, \eqref{eq:A-C-bEII}, and~\eqref{eq:B-D-bEII}, we can conclude
\begin{align}
    H(\bA,\bB &| \bC, \bD, \cEII(\cL), \cS)\nonumber\\ 
    &\leq H(\bA,\bB,\bEII(\cL) | \bC, \bD, \cEII(\cL), \cS)\nonumber\\
    &= H(\bEII(\cL) | \bC, \bD, \cEII(\cL), \cS) \nonumber\\
    &\qquad +  H(\bA,\bB | \bC, \bD, \cEII(\cL), \cS, \bEII(\cL))\nonumber\\
    &\leq H(\bEII(\cL) |  \cEII(\cL)) \nonumber\\
    &\qquad +  H(\bA | \bC,  \bEII(\cL)) +
    H(\bB | \bD,  \bEII(\cL))\nonumber\\
    &=0. \label{eq:3rdTerm_secProof_II}
    \end{align}
Plugging the inequalities in~\eqref{eq:1stTerm_secProof_II}, \eqref{eq:2ndTerm_secProof_II}, and \eqref{eq:3rdTerm_secProof_II} into \eqref{eq:H(J|E,S)}, we arrive at $ H(\key \:|\: \cEII(\cL),\cS) = 0$, which concludes the proof. \hfill $\square$

\section{Proof of Lemma~\ref{prop:fullrank_uxi}}
\label{app:fullRank}
Let $\cL=\{q_1,q_2,\dots,q_\ell\}$ be the set of comprised nodes accessed by the eavesdropper, and denote by $\bxi^{f}$ the repair encoder matrix of node $f$, as defined in~\eqref{eq:xi-def}. Then, the repair encoder matrix $\bxi^{\cL}$ is formed by concatenating matrices ${\{\bxi^f: f\in \cL\}}$, given by 
\begin{align}
    \bxi^\cL=
    \left[
    \begin{array}{c|c|c|c}
        \bxi^{q_1} & \bxi^{q_2} & \cdots & \bxi^{q_\ell}
    \end{array}
    \right]. 
\end{align}
Recall that $\bxi^\cL$ has $\binom{d}{m}$ rows, labeled by subsets of $[d]$ of size $m$. Moreover, $\ubxi^\cL$ consists of the top ${\binom{d}{m}-\binom{d-\ell}{m}}$ rows of $\bxi^\cL$. Since the rows of $\bxi^\cL$ are sorted in lexicographical order, the bottom $\binom{d-\ell}{m}$ rows are exactly the $m$-subsets of ${[\ell+1:d]}$. This means that the top   $\binom{d}{m}-\binom{d-\ell}{m}$ rows to be included in $\ubxi^\cL$ are those whose labels appear in 
\begin{align}
    \bbI := \{\cI: \cX\subseteq [d], |\cI|=m, |\cI|\nsubseteq[\ell+1:d]\},
    \label{eq:def:I}
\end{align}
i.e., those who contain at least one element from $[\ell]$. Thus, the submatrix $\ubxi^\cL$ of $\bxi^\cL$  is given by 
\begin{align}
    \ubxi^\cL= \bxi^\cL(\bbI,:) =
    \left[
    \begin{array}{c|c|c|c}
        \ubxi^{q_1} & \ubxi^{q_2} & \cdots & \ubxi^{q_\ell}
    \end{array}
    \right],
\end{align}
where each submatrix $\ubxi^{q_j}$ consists of the top $\binom{d}{m}-\binom{d-\ell}{m}$ rows of matrix $\bxi^{q_j}$, for $j\in[\ell]$. 

In order to prove that $\ubxi^\cL$ is full-rank, we identify a square submatrix of $\ubxi^\cL$ of size $\binom{d}{m}-\binom{d-\ell}{m}$, and show that it is full-rank. Recall from~\eqref{eq:xi-def}  that the columns of each $\ubxi^{q_j}$ are indexed by subsets of $\cJ\subseteq[1:d]$ of size $m-1$. We label each column of $\ubxi^\cL$ by a pair $\lbl{j,\cJ}$, where $j\in[\ell]$ and $\cJ\subseteq[d]$ with $|\cJ|=m-1$. Thus, column $\lbl{j,\cJ}$ of $\ubxi^\cL$ is indeed column $\cJ$ of $\ubxi^{q_j}$.  Then, we define $\sbxi^\cL$ to the submatrix of~$\ubxi^\cL$, where column $\lbl{j,\cJ}$ appears in $\sbxi^\cL$ if and only if $\cJ\subseteq [j+1:d]$ and $|\cJ|=m-1$, i.e., $\sbxi^\cL := \ubxi^\cL (:, \bbJ) = \bxi^\cL (\bbI, \bbJ)$, where 
\begin{align}
    \bbJ&=\{\lbl{j,\cJ}: j\in [\ell], \cJ\subseteq [j+1:d], |\cJ|=m-1\}.
    \label{eq:def:J}
\end{align}
Note that the number of columns in $\sbxi^\cL$ is given by
\begin{align*}
    |\bbJ| &= \sum_{j=1}^{\ell} \binom{d-j}{m-1}\stackrel{\rm{(a)}}{=} \sum_{t=d-\ell}^{d-1} \binom{t}{m-1} \\
    &=\sum_{t=0}^{d-1} \binom{t}{m-1} - \sum_{t=0}^{d-\ell-1}
    \binom{t}{m-1}\\
    &\stackrel{\rm{(b)}}{=}
    \binom{d}{m} - \binom{d-\ell}{m},
\end{align*}
where in~{\rm (a)} we have $t=d-j$, and~{\rm (b)} follows from the identity $\sum_{i=0}^{a} \binom{i}{b}=\binom{a+1}{b+1}$. This shows that $\sbxi^\cL$ is a square matrix. 

Next, note that for each $\cI\in \bbI$, we have $\min \cI \in [\ell]$ and $\cI\setminus \{\min \cI\} \subseteq [\min \cI:d]$ and  $|\cI\setminus \{\min \cI\}| = |\cI|-1 = m-1$. Therefore, the pair $\lbl{\min \cI, \cI\setminus\{\min \cI\}}$ belongs to $\bbJ$. This implies that there is a homomorphism between $\bbI$ and $\bbJ$, i.e., $\bbJ \equiv\bbI$ and we have $\sbxi^\cL =  \bxi^{\cL}(\bbJ, \bbJ)$. In the following, we use the format given in~\eqref{eq:def:J} to refer to the rows and column labels of $\bxi^{\cL}$. 

Recall that the rank of a matrix is invariant (subject to a sign) under the permutation of its rows and columns. In order to show that $\sbxi^\cL$ is full-rank, we permute the rows and the columns of $\sbxi^\cL$ to obtain a new matrix $\pbxi^\cL$, and then we show that $\det{\pbxi^\cL}\neq 0$.  To this end, we define a new order on the row and column labels in $\bbJ$, and then sort them with respect to the new order. 

\begin{defi}
    For two pairs $\lbl{i,\cI}, \lbl{j,\cJ} \in \bbJ$, we say $\lbl{j,\cJ}$ dominates $\lbl{i,\cI}$ and write
    \[\lbl{i,\cI} \nsl \lbl{j,\cJ},\]
    if either $\cI \setl \cJ$,  or $\cI = \cJ$ and  $i < j$. 
    \label{defn:orderRelation}
\end{defi} 

In the following, we use $\pJ$ to refer to a sequence of all the pairs $\lbl{j,\cJ}$ in $\bbJ$ (see~\eqref{eq:def:J}), which are sorted with respect to~$\nsl$. Subsequently, we define $\pbxi^\cL$ as a permuted version of $\sbxi^{\cL}$ in which all the rows and columns are permuted with respect to~$\nsl$, i.e., $\pbxi^\cL:= \bxi^{\cL} (\pJ, \pJ)$. Hence, the entry of $\pbxi^{\cL}$ at row $\lbl{i,\cI}$ and column $\lbl{j,\cJ}$ is given by \[\pbxi^{\cL}(\lbl{i,\cI},\lbl{j,\cJ}) = \bxi^{q_j}(\{i\}\cup \cI, \cJ).\] Figure~\ref{fig:flowchart_xi} demonstrates the construction of matrix $\pbxi^\cL$, and Example~\ref{ex:pbxi} illustrates one instance of $\pbxi^\cL$. 
\begin{figure*}[htbp]
    \centering
    \includegraphics[width=0.8\textwidth]{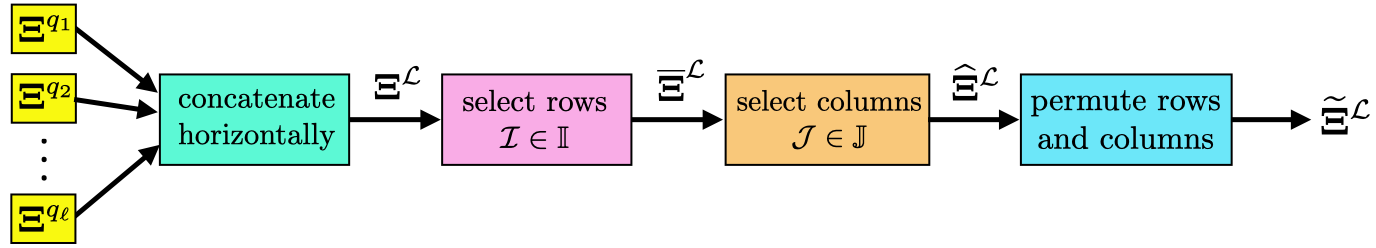}
    \caption{Construction of matrix $\pbxi^\cL$.}
    \label{fig:flowchart_xi}
\end{figure*}

Consider a subset $\cJ\subseteq [d]$ with $|\cJ|=m-1$. Recall from~\eqref{eq:def:J} that every $\lbl{j,\cJ}$ with $j<\min \cJ$ and $j\in [\ell]$ is a column/row label in $\bbJ$. The order $\nsl$ in Definition~\ref{defn:orderRelation} arranges the labels such that all pairs with a common $\cJ$ appear next to each other. This motivates us to define the group of labels associated with each $\cJ\subseteq [d]$ with $|\cJ|=m-1$ as
\begin{align}\label{eq:def:G}
    \cG(\cJ):=\{\lbl{j,\cJ}: j< \min \cJ, j\leq \ell\}.
\end{align}
The following proposition specifies the structure of the matrix $\pbxi^\cL$, and plays a crucial role in the proof of the full-rankness of $\ubxi^\cL$. We present the proof of the proposition at the end of this section. 

\begin{prop}
    The matrix $\pbxi^\cL= \bxi^\cL(\pJ,\pJ)$ is a block lower-triangular matrix. That is, the rows and columns of~$\pbxi^\cL$ can be decomposed into groups ${\{\cG(\cJ): \cJ\subseteq [d], |\cJ|=m-1 \}}$, such that each diagonal block $\bxi^\cL(\cG(\cJ), \cG(\cJ))$ is full-rank, and each block on the right side of each diagonal block is an all-zero matrix. 
    \label{prop:block-diagonal}
\end{prop}
The following example demonstrates the operations we apply on the matrix $\ubxi^{\cL}$ to convert it to a block  lower-triangular matrix. 
\begin{ex}\label{ex:pbxi}
    Consider an $(n,k=6,d=6,\ell=3)$ secure system that operates at mode $m=3$. The repair encoder matrix for each node is a matrix with $\binom{d}{m}=\binom{6}{3}=20$ rows and $\binom{d}{m-1}=\binom{6}{2}=15$ columns. Let $\cL=\{q_1, q_2, q_3\}$ be the set of compromised nodes. Then, $\bxi^{\cL} = [\bxi^{q_1} | \bxi^{q_2} | \bxi^{q_3} ]$ is a $20\times 45$ matrix, and $\ubxi^{\cL}$ is a sub-matrix, including only the top $\binom{d}{m}- \binom{d-\ell}{m}= 20- 1=19$ rows, i.e., all the rows, except the one labeled by $\{4,3,5\}$.  Then, $\sbxi^{\cL}$ will be generated by selecting a subset of columns from each of $\bxi^{q_1}$, $\bxi^{q_2}$ and $\bxi^{q_3}$. As determined in~\eqref{eq:def:J}, the set of columns selected from each $\bxi^{q_j}$ are given by 
    \begin{align*}
        \bxi^{q_1}: & \{2,3\}, \{2,4\}, \{2,5\}, \{2,6\}, \{3,4\}, \{3,5\}, \\
        &\{3,6\}, \{4,5\}, \{4,6\}, \{5,6\}, \\
        \bxi^{q_2}: &  \{3,4\}, \{3,5\}, \{3,6\}, \{4,5\}, \{4,6\}, \{5,6\}, \\
        \bxi^{q_3}: &  \{4,5\}, \{4,6\}, \{5,6\}. 
    \end{align*}
    Hence, we have a total of $10+6+3=19$ columns, and  $\sbxi^{\cL}$ will be a square matrix. Rearrangement of these $19$ columns (as well as the $19$ rows) according to $\nsl$ order provides us with matrix $\pbxi^\cL$ given by
    \begin{align*}
        \includegraphics[width= 0.485\textwidth]{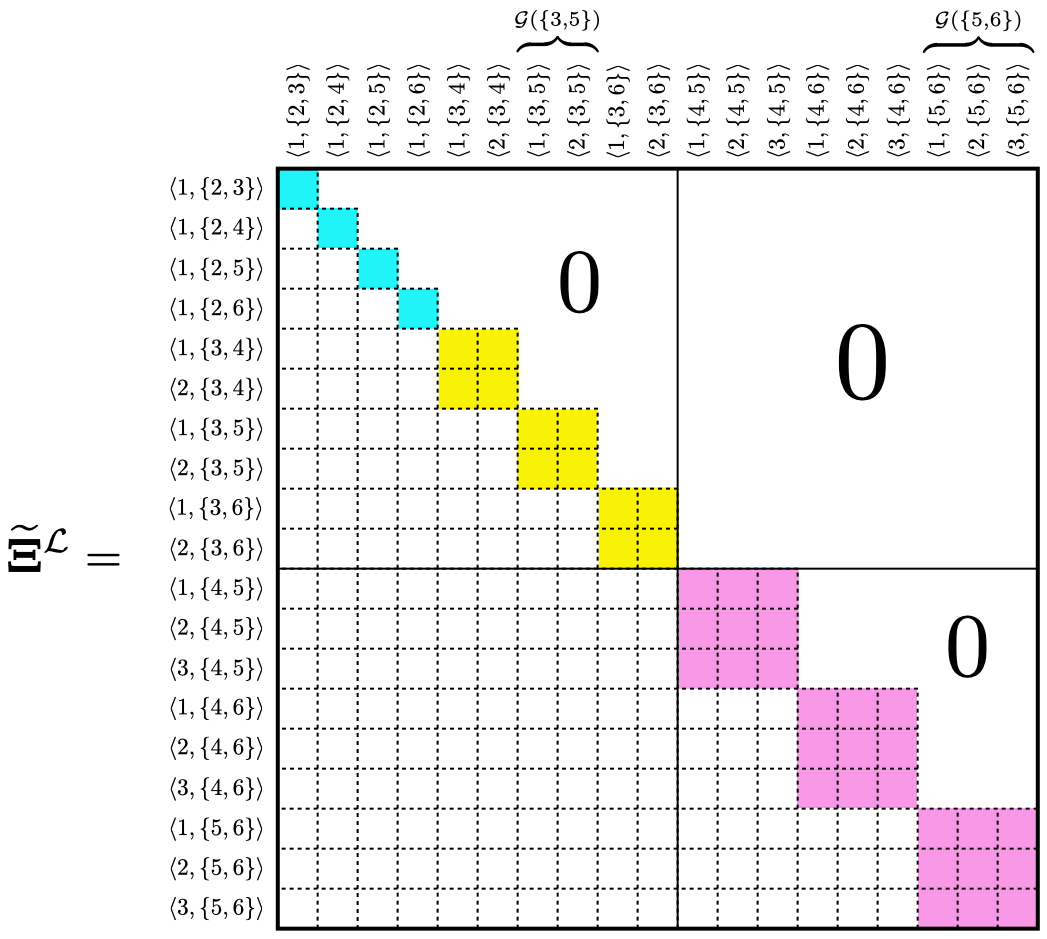}
    \end{align*}
    which is a block lower-triangular matrix consisting of full-rank diagonal blocks. \hfill $\diamond$
\end{ex}

By Proposition~\ref{prop:block-diagonal}, we can write
\begin{align}
    \det{\sbxi^\cL} &\stackrel{\rm{(a)}}{=} \det{\pbxi^\cL}\nonumber \\
    &\!\stackrel{\rm{(b)}}{=} \prod_{\substack{\cJ\subseteq[d] \\ |\cJ|=m-1}} \det{  \pbxi^{\cL}(\cG(\cJ), \cG(\cJ))} \stackrel{\rm{(c)}}{\neq} 0, \nonumber
\end{align}
where~{\rm (a)} holds since the determinant of a matrix is invariant (subject to a sign) to a permutation of its rows and columns, {\rm (b)} follows  from~\cite[Section~0.9.4]{horn2012matrix}, and~{\rm (c)} is a consequence of Proposition~\ref{prop:block-diagonal}. This completes the proof of Lemma~\ref{prop:fullrank_uxi}. \hfill $\square$

\begin{proof}[Proof of Proposition~\ref{prop:block-diagonal}]
First, we show that $\bbJ$ is fully decomposed into the union of the groups $\cG(\cJ)$ over all $\cJ$'s with $\cJ\subseteq [d]$ and $|\cJ|=m-1$. It is clear that $\cG(\cI) \cap \cG(\cJ)=\varnothing$, for distinct $\cI$ and $\cJ$. Next, recall from~\eqref{eq:def:G} that if $\cJ\subseteq [\ell+1:d]$, then ${\cG(\cJ)=\{\lbl{j,\cJ}: j\in [\ell]\}}$ and $|\cG(\cJ)| = \ell$. Similarly, if $\cJ\nsubseteq [\ell+1:d]$, then ${\cG(\cJ)=\{\lbl{j,\cJ}: j \leq \min\cJ -1  \}}$ and $|\cG(\cJ)| = \min \cJ-1$. Note that if $\cJ\nsubseteq [\ell+1:d]$, then it should have at least one element from $\ell$ and hence, we have $x=\min \cJ\in [\ell]$ and $|\cG(\cJ)| = x-1$. Thus, we can write
\begin{align}
    &\sum_{\substack{\cJ\subseteq [d]\\ |\cJ|=m-1}}|\cG(\cJ)| 
    = \sum_{\substack{\cJ\subseteq  [d]\\ |\cJ|=m-1\\ \cJ\subseteq  [\ell+1:d] }}|\cG(\cJ)| + 
    \sum_{\substack{\cJ\subseteq [ d]\\ |\cJ|=m-1\\ \cJ\nsubseteq  [\ell+1:d]}}|\cG(\cJ)|\nonumber\\
    &= \binom{d-\ell}{m-1}\ell + 
    \sum_{x=1}^{\ell}\sum_{\substack{\cJ\subseteq [d]\\ |\cJ|=m-1\\ \min \cJ =x}}|\cG(\cJ)| \nonumber\\
    &= \ell \binom{d-\ell}{m-1} +  \sum_{x=1}^{\ell}  \binom{d-x}{m-2} (x-1)
       \nonumber\\
    &= \ell \binom{d-\ell}{m-1} + \sum_{y=d-\ell}^{d-1}  \binom{y}{m-2} (d
    -(y+1))
     \nonumber\\
    &= \ell \binom{d-\ell}{m-1} \!+\!  d\!\sum_{y=d-\ell}^{d-1} \! \binom{y}{m-2} - \!\sum_{y=d-\ell }^{d-1} (y+1) \binom{y}{m-2}     \nonumber\\
    &= \ell \binom{d-\ell}{m-1} \!+\! d \!\sum_{y=d-\ell }^{d-1} \! \binom{y}{m-2} - (m\hspace{-1pt}-\hspace{-1pt}1)\!\sum_{y=d-\ell}^{d-1}  \binom{y+1}{m-1}      \nonumber\\
    &= \ell \binom{d-\ell}{m-1} + d\lb \binom{d}{m-1} - \binom{d-\ell}{m-1}\rb\nonumber\\
    & \qquad - (m-1) \lb \binom{d+1}{m} - \binom{d-\ell+1}{m}\rb  \nonumber\\
    &\stackrel{\rm{(a)}}{=} \binom{d}{m}-\binom{d-\ell}{m} = |\bbJ|, \label{eq:G-partition}
 \end{align}
where in~{\rm (a)} we have used identities ${\binom{a}{b}=\binom{a-1}{b}+\binom{a-1}{b-1}}$ and ${\binom{a}{b}=\frac{a}{b}\binom{a-1}{b-1}}$. Then, \eqref{eq:G-partition}, together with $\cG(\cI) \cap \cG(\cJ)=\varnothing$, implies $\bigcup_{\substack{\cJ\subseteq[d] \\ |\cJ|=m-1}} \cG(\cJ)= \bbJ$. 
 
Now, consider a diagonal block of $\pbxi^{\cL}$ associated to a set $\cJ$, that is, $\pbxi^{\cL}(\cG(\cJ), \cG(\cJ))$. Recall that each row label in $\cG(\cJ)$ is a pair $\lbl{i,\cJ}$ and each column label in $\cG(\cJ)$ is another pair $\lbl{j,\cJ}$, for some $i,j\leq z\coloneqq\min\{\min \cJ-1, \ell\}$. Then, using~\eqref{eq:xi-def}, entry at position $(\lbl{i,\cJ}, \lbl{j,\cJ})$ is given by 
\begin{align*}
    \pbxi^{\cL}(\lbl{i,\cJ}, \lbl{j,\cJ}) &= \bxi^{q_j} (\{i\}\cup \cJ, \cJ) \\
    &= (-1)^{\ind{\{i\}\cup \cJ}{i}} \mathbf{\Psi}(q_j, i) = - \mathbf{\Psi}(q_j, i),
\end{align*}
where the last equality follows from the fact that ${i< \min \cJ}$, and therefore, $i$ is the smallest entry of $\{i\}\cup\cJ$, that is, ${\ind{\{i\}\cup \cJ}{i}=1}$. This shows that the corresponding block is a submatrix of $\mathbf{\Psi}$ (subject to a negative sign), corresponding to the rows in $\{q_1,\dots, q_{z}\}$ and columns in $\{1,2,\dots, z\}$. Then, Condition~\ref{cond:psi:2} implies that $\pbxi^{\cL}(\cG(\cJ), \cG(\cJ))$ is full-rank. 
 
Next, consider an entry at row $\lbl{i,\cI}$ and column $\lbl{j,\cJ}$ that appears on the right side of a diagonal block $\pbxi^{\cL}(\cG(\cI), \cG(\cI))$. This mean the column label $\lbl{j,\cJ}$ dominates the row label $\lbl{i,\cI}$, that is,  ${\lbl{i,\cI} \nsl\lbl{j,\cJ}}$. Since $\cI \neq \cJ$, we have $\cI \setl \cJ$, and hence, $\min \cI \leq \min  \cJ$. These yield to ${i< \min \cI \leq \min \cJ}$, and thus, $i\notin \cJ$. Therefore, 
\begin{align*}
    \left|(\{i\}\cup \cI )\setminus \cJ\right| = |\{i\}| + |\cI \setminus \cJ| \geq 2.
\end{align*}
This, together with the definition of $\bxi^{q_j}$ in~\eqref{eq:xi-def}, implies $\pbxi^{\cL}(\lbl{i,\cI}, \lbl{j,\cJ}) = \bxi^{q_j} (\{i\}\cup \cI, \cJ)=0$. This completes the proof of Proposition~\ref{prop:block-diagonal}. 
\end{proof}

\bibliographystyle{IEEEtran}
\bibliography{bib_secureDSS}

\end{document}